\documentclass[11pt,final]{article}
\usepackage{graphicx} 
\usepackage[dvipsnames,svgnames]{xcolor}
\usepackage[english]{babel}
\usepackage[a4paper,margin=1in]{geometry}
\usepackage{amssymb,amsmath,amsthm}
\usepackage[colorlinks=true, allcolors=blue]{hyperref}
\usepackage{algorithm}
\usepackage[noend]{algpseudocode}
\usepackage{enumitem}
\usepackage[mathscr]{euscript}
\usepackage{mathtools}
\usepackage{caption}
\usepackage{ifdraft}
\usepackage{lineno}

\usepackage[capitalize, nameinlink]{cleveref}
\usepackage[section]{placeins}
\usepackage{soul}
\crefname{theorem}{Theorem}{Theorems}
\crefname{lemma}{Lemma}{Lemmas}
\crefname{invariant}{Invariant}{Invariants}
\crefname{claim}{Claim}{Claims}
\crefname{observation}{Observation}{Observations}
\crefname{algorithm}{Algorithm}{Algorithms}
\crefname{figure}{Figure}{Figures}

\newtheorem{theorem}{Theorem}[section]
\newtheorem{lemma}[theorem]{Lemma}

\newtheorem{definition}[theorem]{Definition}

\newtheorem{remark}{Remark}

\let\citet\cite
\def\ref#1{\textcolor{red}{[\textbackslash ref is disabled. Please use \textbackslash cref instead. Use \textbackslash cref\{sec:foo,sec:bar\} to reference two things at once.]}}

\def\eps{\varepsilon}

\newcommand{\E}{\mathbf{E}}

\newcommand{\EE}[2][]{%
  \if\relax\detokenize{#1}\relax
    \E\left[#2\right]%
  \else
    \E_{#1}\left[#2\right]%
  \fi
}
\DeclareMathOperator*\poly{poly}
\DeclareMathOperator*\Geom{\textsf{Geom}}
\DeclareMathOperator*\Exp{\textsf{Exp}}
\DeclareMathOperator*\TrunExp{\textsf{TrunExp}}
\DeclareMathOperator*\supp{supp}
\DeclarePairedDelimiter\set{\lbrace}{\rbrace}
\renewcommand\O{\mathcal O}
\newcommand\Bout{B^+}
\newcommand\Bin{B^-}
\newcommand\eout{\delta^+}
\newcommand\ein{\delta^-}
\newcommand{\dist}{\mathrm{dist}}
\newcommand{\mk}{\mathrm{mk}}
\newcommand{\vol}{\mathrm{vol}}

\newcommand{\diam}{\mathrm{diam}}
\newcommand{\ey}{\ex'}
\newcommand{\ex}{\mathcal E}
\DeclarePairedDelimiter\ceil{\lceil}{\rceil}
\def\suc{P}
\def\ved{\mathbf{d}}

\allowdisplaybreaks

\title{Stronger Directed Low-Diameter Decompositions\\ with Sub-Logarithmic Diameter and Separation}
\author{Bernhard Haeupler\thanks{INSAIT, Sofia University ``St. Kliment Ohridski'' and ETH Zurich. \href{mailto:bernhard.haeupler@inf.ethz.ch}{bernhard.haeupler@inf.ethz.ch}. This research was partially funded by the Ministry of Education and Science of Bulgaria (support for INSAIT, part of the Bulgarian National Roadmap for Research Infrastructure) and by the European Research Council (ERC) under the European Union's Horizon 2020 research and innovation program (ERC grant agreement 949272).} \and
Richard Hladík\thanks{ETH Zurich. \href{mailto:rihl@uralyx.cz}{rihl@uralyx.cz}. Partially funded by the European Research Council (ERC) under the European Union's Horizon 2020 research and innovation program (ERC grant agreement 949272).} \and
Shengzhe Wang\thanks{Utrecht University. \href{mailto:s.wang5@uu.nl}{s.wang5@uu.nl} Partially funded by the European Research Council (ERC) under the European Union's Horizon 2020 research and innovation program (ERC grant agreement 949272).} \and
Zhijun Zhang\thanks{INSAIT, Sofia University ``St. Kliment Ohridski''. \href{mailto:zhijun.zhang@insait.ai}{zhijun.zhang@insait.ai}. This research was partially done while at Princeton University and was partially funded by the Ministry of Education and Science of Bulgaria (support for INSAIT, part of the Bulgarian National Roadmap for Research Infrastructure).}}
\date{}

\begin{document}

\maketitle

\begin{abstract}
This paper significantly strengthens directed low-diameter decompositions in several ways. 

We define and give the first results for separated low-diameter decompositions in directed graphs, tighten and generalize probabilistic guarantees, and prove new independence results between (far away) edges. Our results are the first to give meaningful guarantees for decompositions with small diameters $D = \Omega(\log \log n)$ in contrast to the state of the art that only applies to super-logarithmic diameters $D = \omega(\log n)$.

These results transfer several important and widely used aspects of undirected low-diameter decompositions to the directed setting. All our results are algorithmic -- small modifications to two existing directed low-diameter decompositions \cite{BFHL25,Li25} can be used to sample decompositions with our new guarantees in near-linear time $\tilde{O}(m)$.
\end{abstract}

\thispagestyle{empty}
\newpage

\thispagestyle{empty}
\tableofcontents

\newpage

\setcounter{page}{1} 
\section{Introduction}

This paper proves stronger and new guarantees for low-diameter decompositions (LDDs) in directed graphs and shows how state-of-the-art algorithms~\cite{BFHL25,Li25}\footnote{The preprint of \citet{Li25} has since been published as \citet{LMR26}. The LDD algorithm is very similar in both versions. In this work, we refer exclusively to the preprint version, although the analysis should be easily adaptible to the conference version as well.} can be adapted to sample such decompositions in near-linear time.

Low-diameter decompositions aim to randomly cut a graph into clusters of bounded diameter as uniformly and unbiased as possible. A sample from a low-diameter decomposition should result in a decomposition where every cluster behaves, as much as possible, like a uniformly random, representative low-diameter neighborhood.
For example, edges should be cut with probability inversely proportional to the target diameter $D$ by a random decomposition and no edge should be much more likely to be cut than any other edge. More generally, one should expect that every short edge, short path, or small neighborhood has a good probability of being clustered (i.e., contained within a single cluster), where short and small are measured in terms of $D$. Moreover, in a random decomposition, far away parts of a graph should be clustered essentially independently, ensuring for example that any shortest path much longer than $D$ gets cut roughly equally often and not too many clusters just barely graze any given long path. 

Since their introduction in the 80's \cite{Awerbuch85}, many variants/notions of uniformity, and algorithms for LDDs have been developed (network decompositions \cite{AwerbuchGLP89}, neighborhood covers \cite{AwerbuchBCP92}, 
and sparse covers \cite{AwerbuchP90a} are some of the names used). LDDs were motivated from the start by applications and have led to countless applications in a wide variety of topics, areas, and fields over the last 40~years. A small subsample of examples includes applications to graph structures (spanners \cite{Cohen98,HarPeledMO23,FiltserGN24}, hopsets \cite{Cohen00,HaeuplerH0RS24}, etc.), data structures (distance oracle \cite{AwerbuchBCP98,ThorupZ01,HaeuplerLS24}, distance labeling \cite{Peleg00}, etc.), network communications (oblivious routings \cite{PelegU89,BuschLT14,ZuzicGYHS22}, synchronizers \cite{Awerbuch85,AwerbuchPPS92,BuschLT14}, etc.), metric embeddings \cite{Bartal96,KrauthgamerLMN04,Filtser21}, and efficient algorithms for 
shortest-paths \cite{KleinS97,AwerbuchBCP98,Cohen00,BernsteinGS21}, low-stretch trees \cite{AlonKPW95,AbrahamBN08,ElkinEST08,KoutisMP11,AN19,ElkinN19}, and related problems in a wide variety of models (parallel \cite{MillerPX13, BlellochGKMPT14}, distributed \cite{AwerbuchGLP89,AwerbuchBCP92}, (semi-)streaming \cite{AhnGM12,FiltserKM23}, derandomization \cite{AwerbuchGLP89,RozhonG20,RozhonEGH22}, etc.).

More recently, research on fast graph algorithms has pushed more into the much harder setting of directed graphs, and unsurprisingly, directed LDDs have started to show their impact there as well. In particular, \citet{BernsteinNW22} showed how LDDs can be generalized to the directed setting and employed their probabilistic guarantees in a breakthrough result, giving a simple near-linear time combinatorial algorithm for negative-length single-source shortest paths in directed graphs. Follow-up works have refined both the decomposition~\cite{BFHL25,Li25}
and its application in this context~\cite{BringmannCF23,fischer2025simple,HJS25,HJS26}.
These are just the beginning of our understanding of directed LDDs and we expect that in the (near) future, many different versions and refinements of directed LDDs will be developed and find a wide variety of applications, as their undirected versions have over the last 40~years and very much continue to this day. 

This paper makes several steps in this direction by giving positive and negative results on how key aspects of undirected LDDs transfer and work in the directed setting. 
The rest of this paper is organized as follows: We first formally state the simplest undirected LDD guarantee in \cref{sec:intro-uLDD-simple}, which also maps directly to the current state-of-the-art for the directed setting. In \cref{sec:intro-uLDD-general}, we then state many more general and stronger LDD properties that have been developed 
in the undirected setting. These properties are key and crucially used in many applications. 
\Cref{sec:our-results} states our results, which demonstrate that many of these guarantees have equivalents in the directed setting. %
In \Cref{sec:our-techniques}, we summarize our techniques.
In \cref{sec:geom-cutting,sec:jason-modified}, respectively, we present and analyze algorithms adapted from \cite{BFHL25,Li25} to prove and provide efficient sampling algorithms for these new and improved directed LDD guarantees.

\subsection{Simplest Low-Diameter Decompositions -- Undirected and Directed}
\label{sec:intro-uLDD-simple}

In this section, we first give some minimal notation and key definitions.
Then, the simplest version of undirected LDDs is stated.
We will also explain how this corresponds to the current state-of-the-art in the directed setting. 

Consider a graph $G=(V,E)$ with $n=|V|$ vertices and $m=|E|$ edges.
Throughout, we assume edges have positive integral lengths, which naturally induce shortest-path distances $\dist_G(\cdot,\cdot)$. The \emph{diameter} of a vertex subset $U$ is defined as $\diam(U) = \max_{u,v \in U} \dist_G(u,v)$.
A \emph{clustering} is a collection of disjoint vertex subsets $C_{1}, C_{2}, \ldots$ called \emph{clusters}. We say a clustering has diameter $D$ if for every cluster $C_i$, its diameter $\diam(C_i)$ is at most $D$. In the \emph{undirected} setting, an edge $\{u,v\}$ is \emph{clustered} by a clustering if there exists an $i$ such that cluster $C_i$ contains both endpoints, otherwise we say the clustering \emph{cuts} the edge $\{u,v\}$. 

With this minimal setup, we can now formally state the simplest undirected LDDs:

\begin{theorem}[Simplest Undirected LDDs]\label{thm:uLDDsimple}
For any undirected graph $G$ and diameter $D$, there exists a distribution over $D$-diameter clusterings such that the following holds 

for every edge $e$ of length $d_{e}$:
\[
\Pr(\text{$e$ is cut}) \leq  \frac{d_{e}}{D} \cdot \O(\log n).
\]
\end{theorem}

While it is immediately obvious, by looking at a path of unit-length edges, that at least a $\frac{1}{D}$ fraction of edges needs to be cut, implying the multiplicative factor is at least $1$, it is also known~\cite{Bartal96} that a loss of $\Theta(\log n)$ is necessary and the best possible up to a constant. 

\emph{Low-diameter decompositions for directed graphs}, as proposed by~\citet{BernsteinNW22}, are similar in that one cuts edges as uniformly as possible, such that the remaining \emph{strongly connected components} have small diameters. Notably, this allows for edges to go between clusters as long as they do not form cycles or, equivalently, as long as they follow some topological order. The only difference in the formal setup is that we use the ordering of the clusters as a (topological) order and edges are considered cut only if they are against this order.

The best known directed LDDs, stated in the following, are due to \citet{BFHL25,Li25}.

\begin{theorem}[State-of-the-art Directed LDDs]\label{thm:dLDDsimple}
For any directed graph $G$ and diameter $D$, there exists a distribution over ordered $D$-diameter clusterings such that the following holds 
for every edge $e$ of length $d_{e}$:
\[
\Pr(\text{$e$ is cut}) \leq  \frac{d_{e}}{D} \cdot \O(\log n \log \log n).
\]
\end{theorem}

In fact, a factor of $\O(\log^2n)$ loss was initially proved by \citet{BernsteinNW22}, which was later improved by \citet{BFHL25} to the state-of-the-art loss factor $\O(\log n \log \log n)$. \citet{Li25} subsequently gave a different proof for the same loss factor. Given that the directed setting is a strict generalization of the undirected setting (one can replace an undirected edge by two directed edges), the loss factor of \Cref{thm:dLDDsimple} is optimal up to the extra $\log \log n$ factor.
Furthermore, this additional $\log \log n$ factor appears in many other different contexts (see for example \cite{B98,AN19,CLRS20}).
Essentially, the same overhead can all be attributed to a general approach by \citet{S95} in analyzing recursive structures.
It seems very unlikely that current techniques can avoid this.
It is quite possible that the $\log \log n$ loss is necessary and tight in the directed setting.

We remark that all results discussed above are known to have near-linear time implementations. 

\subsection{General Undirected Low-Diameter Decompositions}
\label{sec:intro-uLDD-general}

As mentioned above, much stronger guarantees than \Cref{thm:uLDDsimple} are known for undirected LDDs.
For any path, subset of edges, or ball, it is considered cut if any of its edges is cut.
Next, we state the more general version, part or all of which can be obtained for example from~\cite{LinialS93,AwerbuchBCP92,Bartal96,BlellochGKMPT14,MillerPX13}:

\begin{theorem}[General Undirected LDDs]\label{thm:uLDDgeneral}
Let $\epsilon > 0$ be a sufficiently small constant.
For any undirected graph $G$ and diameter $D$, there exists a distribution over $D$-diameter clusterings such that:
    \begin{enumerate}[noitemsep]
        \item For every edge $e$ of length $d_{e} \le \epsilon D$:
        \[\Pr(\text{$e$ is not cut}) \ge \exp\left(- \frac{d_{e}}{D} \cdot \O(\log n)\right).\]
        \item For every path $P$ of length $d_{P} \le \epsilon D$:
        \[\Pr(\text{$P$ is not cut}) \ge \exp\left(- \frac{d_{P}}{D}\cdot \O(\log n)\right).\]
        \item For every subset of edges $\Gamma$ with $d_\Gamma = \sum_{e \in \Gamma} d_e \le \epsilon D$:
        \[\Pr(\text{$\Gamma$ is not cut}) \ge \exp\left(-\frac{d_{\Gamma}}{D}\cdot \O(\log n)\right).\]
        \item For every vertex $v$ and ball $B(v,d_B)$ centered at $v$ with radius $d_B \le \epsilon D$:
        \[\Pr(\text{$B(v,d_B)$ is not cut}) \ge \exp\left(- \frac{d_B}{D}\cdot \O(\log n)\right).\]
        \item For any two edges, paths, subsets of edges, or balls that are distance $\Omega(D)$ apart, the above probabilities are as if independent.
    \end{enumerate}
Furthermore, there is a near-linear time algorithm that samples such a low-diameter decomposition. 
\end{theorem}

Note that guarantees for edges and paths are direct corollaries of guarantees for subsets of edges and balls.
\Cref{thm:uLDDgeneral} strengthens \Cref{thm:uLDDsimple} in multiple ways as we discuss next:

\paragraph{Stronger Probability Bounds and Clustering with Constant Diameters.}
The first notable difference already occurs for the cutting probabilities of edges. Indeed, the $\exp(-\frac{d}{D} \cdot \O(\log n))$ guarantee for edges being not cut in \cref{thm:uLDDgeneral} is strictly stronger than the $\frac{d}{D} \cdot \O(\log n)$ guarantee for an edge being cut in \cref{thm:uLDDsimple}, because $\exp(-\frac{d}{D} \cdot \O(\log n)) > 1 - \frac{d}{D} \cdot \O(\log n)$. In particular, for $d = \O(\frac{D}{\log n})$, both give at least a constant probability for an edge not to be cut. However, $1 - \frac{d}{D} \cdot \O(\log n)$ becomes vacuous for even larger $d$ while $\exp(-\frac{d}{D} \cdot \O(\log n))$ remains meaningful for $d$ as large as $\epsilon D$ for any sufficiently small (constant) $\epsilon$, in which case with probability $n^{-\O(\epsilon)}$, the edge is not cut. This enables meaningful clustering results with constant-diameter clusters. 

\paragraph{Results for Subsets of Edges and Their (In)dependence.}
The guarantees for edges can be generalized to hold for paths and even arbitrary subsets of edges.
From \cref{thm:uLDDsimple}, one can get a similar bound for subsets of edges, with a linear dependency on $d/D$, for free by union bound. 
However, the stronger bounds in \cref{thm:uLDDgeneral} cannot be derived this way. 
Instead, note that intuitively, if only balls with independent and bounded radii are cut, the events of being cut for two far away edges are as if independent,
i.e., the probability bound still holds for one edge being cut, when we condition on the other far away edge.
We give the formal definition for the independence property later in \cref{thm:geom-modified}.
Indeed, such 
property also holds for far away paths, subsets of edges, and balls as stated in \cref{thm:uLDDgeneral}.
Furthermore, the probability guarantees for them essentially implies that even if they are not far away, these events are never worse than independent (i.e., could either be as if independent or positively dependent).

\paragraph{Results for Neighborhoods.}
Motivated by applications, especially in non-sequential models, there has been a  line of work, tracing back to~\cite{AwerbuchP90a,AwerbuchBCP92,AwerbuchBCP98}, devoted to decompositions preserving not only individual edges but also small neighborhoods.
Formally, \cref{thm:uLDDgeneral} shows that all (potentially exponentially many) paths starting from a node and all vertices in a $d$-neighborhood are all clustered together with similar probabilities. To see why this is useful, note that it is basically equivalent to the following \cref{thm:uLDDseparated} about separated (partial) clusterings. A partial clustering is a clustering in which the union of $C_i$ is not all of $V$. We say a node $v$ is clustered if it is contained in some cluster and unclustered if $v \in V \setminus \bigcup_i C_i$. We remark that \cref{thm:dLDDsimple,thm:uLDDgeneral} do not require clusterings to be a partition.
Nevertheless, since edges of any unclustered vertex are by definition always cut, it is without loss of generality to assume all vertices are clustered.
\cref{thm:uLDDseparated} below examines clustering probabilities of vertices instead of cutting probabilities of edges and cares about separation between clusters.

\begin{theorem}[Undirected LDDs with Separation]\label{thm:uLDDseparated}
For any undirected graph $G$, diameter $D$, and separation $d$, there exists a distribution over partial $D$-diameter clusterings such that:
    \begin{enumerate}[noitemsep]
        \item For every vertex $v$:
         \[\Pr(\text{$v$ is clustered}) \ge \exp\left(-\frac{d}{D}\cdot \O(\log n)\right).\] 
        \item For any two clusters $C,C'$, and any two vertices $u \in C, v \in C'$, it holds that $\dist_G(u,v) > d$.
    \end{enumerate}
\end{theorem}

Up to constants, \cref{thm:uLDDseparated} is actually equivalent to \Cref{thm:uLDDgeneral}.
To see this, fix the value of $d$.
Given LDDs where neighborhoods are preserved by \cref{thm:uLDDgeneral}, we consider each vertex $v$ as clustered if its ball $B(v,d/2)$ is not cut.
Then, for any two clustered vertices $u,v$ from different clusters, $B(u,d/2)$ and $B(v,d/2)$ are disjoint.
Consequently, clusters are at least distance $d$ apart from each other.
Conversely, suppose we are given separated clusterings by \cref{thm:uLDDseparated}.
We assign each unclustered vertex $v$ to the nearest cluster.
Due to the separation, $v$ can only be in the ball $B(u,d/2)$ for originally clustered vertices $u$ in at most one cluster.
Moreover, if such a cluster exists, $v$ is assigned to this cluster.
As a result, the ball $B(u,d/2)$ is never cut for any originally clustered vertex $u$.
For both directions, neighborhood preservation probabilities translate to vertex clustering probabilities and vice versa.

\subsection{Importance of Decompositions with Small Diameter and Separation}

It is important to note that in many applications of (undirected) low-diameter decompositions, the stronger guarantees of \Cref{thm:uLDDseparated} and \Cref{thm:uLDDgeneral} compared to the simple per edge $1 - \frac{d}{D} \cdot \O(\log n)$ cutting guarantee of \Cref{thm:uLDDsimple} are absolutely crucial. 

Indeed, in many cases $\frac{d}{D}$ corresponds to a stretch, length-slack, or approximation ratio while the clustering or cutting probability goes into the running time or size and being able to achieve sub-logarithmic, typically even constant, slack is key.

Examples among the many listed above include: Distance oracles and distance labeling schemes, which are data structures that provide $\O(k)$-approximate distances while storing only $\tilde{\O}(n^{1+1/k})$ bits of information about a graph. They are most interesting for $k$ being a small constant or at least sub-logarithmic and they crucially rely on clusterings with diameter proportional to $k$. Also intimately connected to LDDs are graph spanners. A $(2k-1)$-stretch spanner for a given graph $G$ is a sparse subgraph in which all distances are approximated up to the stretch-factor. For every graph there exists such a subgraph with only $\O(n^{1 + 1/k})$ edges. LDDs with diameter $D=k$ (for unweighted graphs with $d=1$) have been key from the start in both proving the existence and giving fast (parallel and distributed)  constructions of spanners and closely-related synchronizers~\cite{Awerbuch85} (see for example the paper ``Improved parallel algorithms for spanners and hopsets''~\cite{miller2015improved}). 

Separation and independence are also used in many applications. They lead for example to neighborhood covers with constant length-approximation or length-slack $\frac{d}{D}$. Together they for example guarantee that running an LDD with diameter $D$ and contracting the components shrinks the diameter of the graph by at least $d$ with high probability (see for example \cite{haeupler2016faster,czumaj2021exploiting}). Even the optimal~\cite{abboud2022hardness} way to compute $t$-pair approximate shortest paths that runs in $(m+t)\cdot n^{\O(1/k)}$ time simply uses $n^{\O(1/k)}$ LDDs for every scale $D=2^i$, checks at what scale a pair appears in a joint cluster, and uses the fact that paths of length $d = D/k$ are clustered with good enough probability to obtain an $\O(k)$-approximation. 

As a last of many examples, all recent works on length-constrained expanders (see for example \cite{haeupler2025cut}) and their applications (such as $\O_\eps(1)$-approximate dynamic distance oracles~\cite{HaeuplerLS24} with $n^{1/\eps}$ update time or $\O_\eps(1)$-approximate $t$-commodity flows in $(m+t)^{1+1/\eps}$ time~\cite{HaeuplerH0RS24} crucially rely on low-diameter decompositions with sub-logarithmic stretch.

\section{Results and Technical Overview}

\subsection{Our Results}\label{sec:our-results}
Our main contribution is the following result on directed LDDs that have the same strong guarantees as in \cref{thm:uLDDgeneral} for undirected LDDs, except for the hard-to-avoid $\log\log n$ factor.

\begin{theorem}[General Directed LDDs]\label{thm:dLDDgeneral}
Let $\epsilon > 0$ be a sufficiently small constant.
For any directed graph $G$ and diameter $D$, there exists a distribution over ordered $D$-diameter clusterings such that:
\begin{enumerate}[noitemsep]
    \item For every edge $e$ of length $d_{e} \le \epsilon D/\log\log n$:
        \[\Pr(\text{$e$ is not cut}) \ge \exp\left(-\frac{d_{e}}{D}\cdot \O(\log n \log\log n)\right).\]
    \item For every path $P$ of length $d_{P} \le \epsilon D/\log\log n$:
        \[\Pr(\text{$P$ is not cut}) \ge \exp\left(-\frac{d_{P}}{D}\cdot \O(\log n \log\log n)\right).\]
    \item For every subset of edges $\Gamma$ with $d_\Gamma = \sum_{e \in \Gamma} d_e \le \epsilon D/\log\log n$:
        \[
        \Pr(\text{$\Gamma$ is not cut}) \ge \exp\left(-\frac{d_\Gamma}{D}\cdot \O(\log n \log\log n)\right).
        \]
    \item For any two edges, paths, or subsets of edges that are distance $\Omega(D)$ apart \emph{in the underlying undirected graph} (i.e., the graph obtained by replacing each directed edge with an undirected edge of the same length), the above probabilities are as if independent.
\end{enumerate}
Furthermore, there is a near-linear time algorithm that samples such a low-diameter decomposition.
\end{theorem}

The strong guarantees of \cref{thm:dLDDgeneral} enable meaningful discussion about directed LDDs with substantially sub-logarithmic diameters as small as $\Theta(\log \log n)$.

We remark that compared to \cref{thm:uLDDgeneral}, \cref{thm:dLDDgeneral} lacks the property for neighborhoods.
As will be seen in \cref{sec:our-techniques}, it turns out that it is impossible to preserve out-balls/in-balls for each vertex with nonzero probability in the directed setting.
This also implies that we no longer obtain LDDs with separation for free.
Indeed, it is, a priori, even unclear if directed LDDs with any notion of separation should exist at all.
The recursive approaches taken by all known directed LDD constructions may increase distances between vertices when recursing into induced subgraphs.
It thus demands a careful selection of clustered vertices to ensure that separation in induced subgraphs imply separation in the original graph. 

Nevertheless, while there does not seem to be a truly directed analog for neighborhoods, we do prove the following meaningful generalization of \Cref{thm:uLDDseparated} to the directed setting:

\begin{theorem}[Directed LDDs with Separation]\label{thm:dLDDseparated} 
For any directed graph $G$, diameter $D$, and separation $d$, there exists a distribution over ordered partial $D$-diameter clusterings such that:
\begin{enumerate}[noitemsep]
    \item For every vertex $v$:
        \[
        \Pr(\text{$v$ is clustered}) \ge \exp\left(- \frac{d}{D}\cdot \O(\log n \log\log n)\right).
        \]
    \item For any two clusters $C,C'$ such that $C$ is ordered after $C'$, and any two vertices $u \in C, v \in C'$, it holds that $\dist_G(u,v) > d$.
\end{enumerate}
Furthermore, there is a near-linear time algorithm that samples such a low-diameter decomposition.
\end{theorem}

This allows to show meaningful separation properties for partial clusterings with diameters as small as $\Theta(\log \log n)$.

\subsection{Our Techniques}\label{sec:our-techniques}

Next we give a technical overview of our directed LDDs.

\paragraph{General Directed LDDs.}
\cref{thm:dLDDgeneral} is proved by adapting \cite{BFHL25}'s algorithm.
We first give a brief overview of \cite{BFHL25}'s algorithm and its analysis, before introducing our modifications and new proof ideas.

At a high level, the algorithm of \cite{BFHL25} repeatedly cuts balls of radii between $r_0$ and $r_1$, both of which are initially set to be $\Theta(D)$.
Each time, it picks a random center vertex $x$ and samples a radius $r \sim r_0+\Geom(p)$ for some $p$ such that $r \ge r_1$ (implying the algorithm fails) only occurs with probability $\delta=1/\poly\log(n)$ across all cutting steps.
If $r \ge r_1$ does occur, the entire algorithm restarts from the very beginning.
Otherwise, by cutting the radius-$r$ ball around $x$, the algorithm recurses on the subgraph induced by the ball while continuing the cutting procedure on the remaining graph.
Once there are no such balls, it gradually decreases $r_0,r_1$ (so cutting smaller balls) and adjusts $p$ accordingly.
The analysis of (single edge) cutting probabilities crucially relies on the following arguments.
\begin{enumerate}
    \item Since $\delta$ is very small, restarting the algorithm can at most double the cutting probability.
    \item When no restarting happens, the overall cutting probability is upper bounded by a union bound over all cutting steps of individual balls.
    \item For any edge $e=(u,v)$ of length $d_e$ and any ball $B$, by the memoryless property of geometric distributions, it is shown that conditioned on $u$ being included in $B$, $v$ is not included in $B$ (so $e$ is cut by $B$) with probability at most $\O(d_e p)$.
\end{enumerate}
However, none of these arguments is suitable for obtaining our stronger exponential-type bounds.
\begin{enumerate}
    \item A restarting probability of $\delta=1/\poly\log(n)$ is no longer affordable as it can ruin any exponential-type bound we get, which can be as small as $n^{-\O(\epsilon)}$ for some $\epsilon$.
    \item Even if exponential-type bounds can be shown for individual cutting steps using prior analysis, they cannot be straightforwardly combined using union bounds.
        It will not produce exponential-type bounds for overall cutting probabilities.
    \item It turns out that proving exponential-type bounds for individual cutting steps is also problematic.
        Indeed, consider some radius-$r$ ball $B$ around some vertex $x$ and some edge $e=(u,v)$ such that $\dist(x,u) \in [r_1-d_e,r_1)$.
        Conditioned on $u$ included in $B$, $v$ is actually never included in $B$ as $r \in [\dist(x,u),\dist(x,u)+d_e)$.\footnote{$r \ge \dist(x,u)$ because $u$ is included in $B$, and $r < r_1 \le \dist(x,u)+d_e$ because no restarting is allowed anymore due to the first issue.}
        That is, $e$ is always cut.
        Recall that the ultimate goal is to show some small, but strictly positive, probability of not cutting each edge.
\end{enumerate}

Our modified algorithm resolves the first issue by utilizing truncated exponential distributions instead of geometric distributions\footnote{It is the truncation that eliminates restarting. The choice of exponential distributions over geometric distributions is mainly for technical reasons as geometric distributions are not well defined for $p>1$. Other than that, exponential distributions can be viewed as the analog of geometric distributions in the continuous setting, and they both have the memoryless property, which is the most important for the sake of analysis.}.
Regarding the second issue, independence between the cutting probabilities we choose to analyze is critical for enabling us to multiply them together.
To address this as well as the third issue, we propose a new dynamic-programming-style approach in \cref{sec:geom-cutting}.
It completely eliminates the need to analyze conditional cutting probabilities as in \cite{BFHL25}, by a more careful analysis of how the value of $r$, relative to $r_0,r_1$, affects the cutting probabilities.
Intuitively, the scenario above in the third issue should not hurt too much as it implies $r \in [r_1-d_e,r_1)$, which only happens with very small probability because $r_1-r_0 \gg d_e$ and tail probabilities of exponential distributions decay super fast.

As an added bonus, our approach also enjoys the benefit of easy generalization to the case of paths and subsets of edges, with little extra effort.
This is not possible with \cite{BFHL25}'s analysis.
Actually, even one path may end in multiple segments due to cutting a single ball.
It is unclear if there is an easy way to deal with the complicated conditional probabilities arising from this.
In contrast, our proof in \cref{sec:geom-cutting} directly works with the case of subsets of edges and the other two follow as direct corollaries.
We are able to show that there are essentially only two situations to consider for each cutting step, regardless of the subset of edges.
Specifically, we can assume without loss of generality that either all edges are included in the ball and go into recursion, or all edges are outside the ball.
Any other partition of edges into the two parts will never lead to larger cutting probabilities.

We also remark that the independence stated in \cref{thm:dLDDseparated} holds so long as the decomposition is obtained by cutting balls with independently sampled diameters at most $\O(D)$.
This is because for each ball, at least one of the two subsets of edges can never be cut, regardless of the radius, due to the distance assumption.
The claimed independence then follows from the independence between sampled radii.
So far, many known LDD results, including \cite{BernsteinNW22,BringmannCF23,BFHL25}, follow this generic approach, while \cite{Li25} is one notable exception.

\paragraph{No Directed Neighborhoods.}

Regarding neighborhoods, it is not clear what is the proper notion of neighborhood for a vertex in the directed setting. One tempting option would be the union of its out-ball and its in-ball. Alternatively, we may consider the intersection of its out-ball and its in-ball. It turns out that neither of the two definitions makes much sense. Indeed, it is even impossible for every vertex to have a strictly positive probability of its out-ball (or in-ball) not being cut, because of the following promised counterexample, which shows the requirement of not cutting an out-ball may be in conflict with the requirement of having a small diameter.
\begin{lemma}
There exists a strongly connected graph with diameter $n-1$ such that all vertices have distance at most $1$ from some vertex.
\end{lemma}
\begin{proof}
Consider a directed cycle $(v_1,v_2,\ldots,v_n)$ and add additional edges from $v_1$ to all other vertices.
Clearly, $v_1$ reaches every vertex using at most one edge.
The diameter of the constructed graph is $n-1$ because the only path from $v_2$ to $v_1$ goes around the entire cycle.
\end{proof}

Instead, suppose one only cares about the intersection of the $d$-radius in-ball and the $d$-radius out-ball.
Note that the intersection is a strongly connected component with diameter at most $2d$.
The desired property is already implied by \cref{thm:uLDDgeneral} since we can consider the round-trip distance $\dist_G(u,v) + \dist_G(v,u)$, which is a symmetric metric and thus undirected LDDs are sufficient.
This has for example been observed and used in \cite{HaeuplerLSW25}.

From the round-trip distance, one can also get a partial clustering in which any two clusters are ``separated'' in the following weak sense: For any two clusters $C, C'$, and any two vertices $u \in C$ and $v \in C'$, either $\dist_G(u,v)>d$ or $\dist_G(v,u)>d$.
However, such guarantee is not of much use because clusters as a whole are not meaningfully separated in either direction. Indeed, there can be vertices $u,u' \in C$ and $v,v' \in C'$ such that $\dist_G(u,v) \ll d$ but $\dist_G(v,u) > d$, while $\dist_G(u',v') > d$ and $\dist_G(v',u') \ll d$.

\paragraph{Directed LDDs with Separation.}

Roughly speaking, to achieve separation in \cref{thm:dLDDseparated}, each time we execute a cut, any vertex within distance $d$ (in the direction of cutting) to the boundary is marked as unclustered.
Crucially, vertex marking is needed on both sides of the boundary to ensure that any short path between unmarked vertices are entirely contained in induced subgraphs after cutting (\cref{lem:dist}).
In \cref{sec:jason-modified}, we adapt and simplify \cite{Li25}'s algorithm to prove \cref{thm:dLDDseparated}.
The construction of \cite{Li25}, based on \cite{CalinescuKR04}, samples a single cutting radius, which is used for all randomly ordered balls in the same iteration, from a uniform distribution.
Due to this nature of the cutting radius, it may be quite surprising that any exponential-type bound would even hold in the first place.

At a high level, the essence of \cite{Li25}'s proof is showing that even though the number $k$ of balls to cut in each iteration can be huge, randomly ordering all balls ensures that the probability of not clustering\footnote{Technically, \cite{Li25}'s proof is for edge cutting, but it can be adapted to clustering with separation using our modification to the algorithm.} any vertex $v$ only grows linearly in $\log k$.
Building upon this, we further observe that to achieve separation $d$, relevant balls must satisfy $\dist(x,v) \in (r_0-d,r_1+d]$ where $x$ is the center of the ball and $(r_0,r_1]$ is the range from which the cutting radius is sampled.
With $k \gg (r_1+d)-(r_0-d)$ balls, each of which has $\dist(x,v)$ be an integer in the range $(r_0-d,r_1+d]$, we are actually able to prove a much tighter bound of the form $1-\exp(-\O(\log k))$ using convexity arguments (\cref{lem:prob-mark}).

\paragraph{Remark on \cite{BFHL25} vs.~\cite{Li25}.}
For the modified \cite{Li25}'s algorithm, we will show that the same probability guarantee as in \cref{thm:dLDDgeneral} also holds for edges, but it is not clear if this is the case for paths or subsets of edges.
In fact, by using the same radius, sampled from a uniform distribution, for all balls in each iteration, the events of being cut may be correlated in complicated ways even for edges that are further apart than $\Omega(D)$, as opposed to the independence stated in \cref{thm:dLDDgeneral}, which is satisfied by the modified \cite{BFHL25}'s algorithm.

We also remark that our modified version of \cite{BFHL25}'s algorithm can be adapted, by the same vertex marking procedure, to provide the separation guarantee as in \cref{thm:dLDDseparated}.
Vertex clustering probabilities can be bounded by a similar argument to the one for cutting probabilities of subsets of edges (each cutting step has $2d$ ``bad'' radii instead of $d_\Gamma$).
Regarding the running time, one has to be careful as vertices are marked on both sides of the boundary.
In other words, not all marked vertices go into recursion.
We have to bound the total contribution of remaining vertices to the running time.
Fortunately, since centers of balls are sampled randomly, a similar argument to the one for the modified \cite{Li25}'s algorithm can show the running time is still subsumed by that of Cohen's algorithm \cite{C97} with high probability.
We choose to present the analysis for the modified \cite{Li25}'s algorithm, as the same type of analysis is also crucial to proving probability bounds there.

\section{Preliminaries}\label{sec:prelims}

We consider directed simple graphs $G = (V, E)$, where $E \subseteq V^2$, with $n = |V|$ vertices and $m = |E|$ edges.
For each vertex $v \in V$, its \emph{degree} $\deg(v)$ is the total number of edges with $v$ being one of the endpoints.
For any $U \subseteq V$, its \emph{volume} $\vol(U)$ is defined as $\sum_{v \in U} \deg(v)$.
Note that $\vol(V)=2m$.
We write $\overline U = V \setminus U$. 
Let $G[U]$ be the subgraph induced by $U$. 
We denote with $\eout(U)$ the set of edges that have their starting point in $U$ and endpoint in~$\overline U$. 
We define $\ein(U)$ symmetrically. 

For each edge $e \in E$, we write $d_e$ for its length.
The distance from vertex $u$ to $v$ is written $\dist_G(u,v)$.
For any $U \subseteq V$, we say that $U$ has \emph{diameter} $D$ if for all pairs of vertices $u, v \in U$, we have $\dist_G(u, v), \dist_G(v, u) \leq D$. A strongly connected component in a directed graph $G$ is a subgraph where for every pair of vertices $u,v$ in the component, there is a path from $u$ to $v$ and vice versa. For a radius $r \geq 0$, the \emph{out-ball} of vertex $v$ is $\Bout_G(v, r) = \set{x \in V \mid \dist_G(v, x) \leq r}$ and the \emph{in-ball} is $\Bin_G(v, r) = \set{y \in V \mid \dist_G(y, v) \leq r}$.
We may omit the subscript if it is clear from context.

\section{Strong Bounds for Modified \cite{BFHL25}'s Algorithm}
\label{sec:geom-cutting}

In this section, we adapt and strengthen \cite{BFHL25} in two ways, proving \cref{thm:dLDDgeneral}. We improve the bound on the probability that an edge is cut from $\O(d_e / D \cdot \log n \log \log n)$ to $1 - \exp(-\O(d_e / D \cdot \log n \log \log n))$, which is stronger when $d_e$ is close to $D$. We also show that analogous bounds even hold for paths and subsets of edges in general.
The following is a more formal version of \cref{thm:dLDDgeneral}.

\begin{theorem}
\label{thm:geom-modified}
Let $\epsilon > 0$ be a sufficiently small constant.
There is a randomized algorithm that, given any directed graph $G=(V,E)$ with $n$ vertices and $m$ edges, and diameter parameter $D=\poly(n)$, returns an ordered (full) clustering $\sigma=(C_1,\ldots,C_k)$ of $V$, satisfying all of the following.
\begin{description}[labelindent=.5cm]
    \item[Edge cutting probability:] Each edge $e=(u,v)$ is \emph{cut}, i.e., $u \in C_i$ and $v \in C_j$ for $i > j$, with probability at most $1-\exp(-d_e/D \cdot \O(\log n \log\log n))$ if $d_e \le \epsilon D / \log\log n$.
    \item[Path cutting probability:] Each path $P$ is \emph{cut}, i.e., some of its edges are cut, with probability at most $1-\exp(-d_P/D \cdot \O(\log n \log\log n))$ if $d_P \le \epsilon D / \log\log n$.
    \item[Subset cutting probability:] Each subset of edges $\Gamma$ is \emph{cut}, i.e., some of its edges are cut, with probability at most $1-\exp(-d_\Gamma/D \cdot \O(\log n \log\log n))$ if $d_\Gamma \le \epsilon D / \log\log n$.
    \item[Independence property:] 
    For any subsets of edges $\Gamma,X,Y$ such that for every vertex $v$, either $B^\pm_G(v,D) \cap \Gamma$ or $B^\pm_G(v,D) \cap (X \cup Y)$ is empty,
    the above bound for $\Pr(\text{$\Gamma$ is cut})$ also holds for $\Pr(\text{$\Gamma$ is cut}\mid \text{all of $X$ is cut $\wedge$ none of $Y$ is cut})$.
    \item[Diameter property:] Each $C_i$ is strongly connected and has diameter at most $D$.
    \item[Running time:] The algorithm runs in $\O(m \log^5 n \log\log n)$ time with high probability.
\end{description}
\end{theorem}

\subsection{Overview of \cite{BFHL25}'s Algorithm and Our Modification}

Before going into our modification, let us quickly summarize the algorithm of \cite{BFHL25} and describe the slight modifications we need to make the strengthened analysis work. In the following, we ignore many subtle computational and other issues irrelevant to our analysis and focus more on intuition. A more detailed description can be found in \cref{alg:cutting-lemma,alg:bernhard-modified}.

\paragraph{Algorithm summary.}
The algorithm of \cite{BFHL25} repeatedly picks a vertex $v$, samples a radius $r$ and then carves out the ball $B = \Bout(v, r)$ from the graph, by removing its vertices from the vertex set and cutting the out-edge set $\eout(B)$. Then it proceeds by carving more balls in the remaining graph, while also calling itself recursively on $B$.

Logistically, the algorithm is split into two parts: the \emph{cutting lemma} (Lemma 6.2 of \cite{BFHL25}) repeatedly carves out balls of a certain kind (specified by the caller) from $G$, until no such ball remains with high probability. It returns the set of cut edges and the union of all balls that have been carved. The main part of the algorithm repeatedly calls the cutting lemma on the remaining graph with different parameters. It also cleans up the returned graph and runs itself recursively on the returned balls.

The analysis of the edge cutting probability can be roughly split into two parts: first, one bounds the edge cutting probability in one invocation of the cutting lemma. Afterwards, one bounds the edge cutting probability over all the invocations of the cutting lemma that an edge participates in throughout the algorithm.

The most important part of the algorithm is how $v$ and $r$ are chosen in the cutting lemma. The following description is sufficient for our purposes. The cutting lemma receives two parameters, $r_0$ and $r_1$. Each ball center $v$ is sampled at random from the set of ``viable'' center vertices, then $r$ is sampled as $r \sim r_0 + \Geom(p)$, where $p$ is chosen such that $r_0 \le r \le r_1$ with a sufficiently high probability. Whenever $r > r_1$, this invocation of the cutting lemma fails, and the whole main algorithm restarts.

\paragraph{Our change.}
Our only change to the algorithm is as follows: we change $p \gets 2p$ and sample $r \sim r_0 + \lfloor\Exp(p)\rfloor$ with rejection until $r < r_1$.
In other words, we sample from a truncated exponential distribution.
The cutting procedure never fails and the main algorithm never restarts.

\begin{algorithm}[t]
\caption{Modified implementation of the cutting procedure of \cite{BFHL25}.}
\label{alg:cutting-lemma}
\begin{enumerate}[label=\arabic*.]
    \item Run Cohen's algorithm (with parameter $\epsilon = \frac18$) on $G$ to compute approximations~$b_1(v)$ satisfying that $\tfrac78 \cdot |\Bout(v, r_1)| \leq b_1(v) \leq \tfrac98 \cdot |\Bout(v, r_1)|$. Mark all nodes $v$ satisfying that $b_1(v) \leq \frac98 \cdot \frac{m}{s_1}$ as \emph{good} and all other nodes as \emph{bad.}
    \item Repeat the following steps $\log n + 1$ times:
    \begin{enumerate}[label*=\arabic*.]
        \item Repeat the following steps $s_0 \cdot 100 \log n$ times:
        \begin{enumerate}[label*=\arabic*.]
            \item Among the nodes that are marked good, pick a uniformly random node $v$.
            \item Test whether $|\Bout(v, r_0)| \geq \frac{1}{2} \cdot \frac{m}{s_0}$, and otherwise continue with the next iteration of 2.1.
            \item \textcolor{red}{Repeatedly sample $r \sim r_0 + \lfloor\Exp(p)\rfloor$} (where~\smash{$p = \frac{\textcolor{red}2\ln(2s_0 / \delta)}{r_1 - r_0}$}) \textcolor{red}{until $r < r_1$}.
            \item Compute $B^+(v, r)$. Update $S \gets S \cup \eout(B^+(v, r))$, mark the nodes in $B^+(v, r)$ as bad, and remove these nodes from $G$.
        \end{enumerate}
        \item Run Cohen's algorithm (with parameter $\epsilon = \frac18$) on $G$ to compute approximations $b_0(v)$ satisfying that $\tfrac78 \cdot |\Bout(v, r_0)| \leq b_0(v) \leq \tfrac98 \cdot |\Bout(v, r_0)|$. Mark all nodes $v$ satisfying that $b_0(v) < \frac78 \cdot \frac{m}{s_0}$ as bad.
    \end{enumerate}
    \item Let $R$ denote the set of remaining nodes in $G$, and return $(S, R)$.
\end{enumerate}
\end{algorithm}

The modified cutting procedure is given in \cref{alg:cutting-lemma}.
For completeness, we also restate the main algorithm in \cref{alg:bernhard-modified}. For the rest of this section, we fix the following parameters, which are the same as in \cite{BFHL25}:

\begin{itemize}
    \item $L = \ceil{\log\log m} + 1$,
    \item \smash{$\delta = \frac{1}{\log^{10} m}$},
    \item $r_0 := 0$ and \smash{$r_\ell := r_{\ell-1} + \frac{D}{2^{L-\ell+3}} + \frac{D}{4 L}$} (for $1 \leq \ell \leq L$),
    \item \smash{$s_\ell := \min(2^{2^{L-\ell}}, m+1)$} (for $0 \leq \ell \leq L$).
\end{itemize}

We also analyze the algorithm via a different approach, albeit with a similar general idea. Instead of analyzing the cutting lemma and the main algorithm separately, we establish a single recursion that deals with all invocations of the cutting lemma. 
As a bonus, we manage to get rid of the need to analyze conditional probabilities.

We remark that for a slightly easier analysis, during the recursive calls of the algorithm where we have $m' < m$ edges, we keep using the same value of $\delta = (\log m)^{-10}$ instead of changing it to $\delta' = (\log m')^{-10}$. We set the parameters of our algorithm such that all failure events happen with $1/\poly(m)$ probability, instead of $1/\poly(m')$. This only changes some $\log m'$ factors to $\log m$ in the time complexity of the algorithm. Concretely, we need to explicitly use the initial $m$ and $n$ in two places: for Cohen's algorithm, and in step 2.1 of \cref{alg:cutting-lemma}, so that they all succeed ``with high probability'' $1-1/\poly(m)$.

\begin{algorithm}[t]
\caption{The directed LDD algorithm of \cite{BFHL25}.} \label{alg:bernhard-modified}
\begin{enumerate}
    \item Initially let $S \subseteq E$ be the set of edges of length at least $\frac{D}{4L}$, and remove these edges from $G$
    \item For $\ell \gets L, \dots, 1$:
    \begin{enumerate}
        \item[2.1.] Run \cref{alg:cutting-lemma} on $G$ with parameters $\delta, r_{\ell-1}, r_\ell, s_{\ell-1}, s_{\ell}$, and let $S^+_\ell, R^+_\ell$ denote the resulting sets.
        \item[2.2.] Compute the strongly connected components in $(G \setminus S^+_\ell)[V \setminus R^+_\ell]$. Recur on each such component and add the recursively computed cut edges to $S$.
        \item[2.3.] Update $G \gets G[R^+_\ell]$ and $S \gets S \cup S^+_\ell$.
        \item[2.4.] Run \cref{alg:cutting-lemma} on $\mathrm{rev}(G)$ with parameters $\delta, r_{\ell-1}, r_\ell, s_{\ell-1}, s_{\ell}$, and let $S^-_\ell, R^-_\ell$ denote the resulting sets.
        \item[2.5.] Compute the strongly connected components in $(G \setminus S^-_\ell)[V \setminus R^-_\ell]$. Recur on each such component and add the recursively computed cut edges to $S$.
        \item[2.6.] Update $G \gets G[R^-_\ell]$ and $S \gets S \cup S^-_\ell$.
    \end{enumerate}
\end{enumerate}
\end{algorithm}

Before focusing on the cutting probability, let us check that the rest of the original analysis still applies after our modification, including its correctness and its running time complexity. For this we need to look at both changes that we do:
\begin{itemize}
\item Changing the way $r$ is sampled. Clearly, the absence of restarts cannot increase the algorithm's time complexity.
All the analysis except for the edge cutting probability does not rely on how $r$ is sampled, except that $r \in [r_{\ell - 1}, r_\ell)$, which it still is after our change.
\item Changing some factors from $\log n'$ and $\log m'$ (where $n'$ and $m'$ are the number of nodes and edges in the current recursive call) to $\log n$ and $\log m$ (where $n$ and $m$ are the number of nodes and edges in the original graph). This does not worsen the time complexity bound, as the original analysis anyway eventually has $\log m$ and $\log n$ factors.
\end{itemize}

More details can be found in \cite{BFHL25}.

\subsection{Analysis of Modified \cite{BFHL25}'s Algorithm}

We start the analysis by investigating the truncated exponential distribution $r$ is sampled from:
\begin{definition}
For $p > 0$, and integers $a < b$, the truncated exponential distribution $\TrunExp(p, a, b)$ is a distribution defined on $x \in \{a, a + 1, \ldots, b - 1\}$ as follows:
\[
\Pr_{\mathbf x \sim \TrunExp(p, a, b)}(\mathbf x = x) = \frac{\exp(-p(x - a)) - \exp(-p(x - a + 1))}{1 - \exp(-p(b - a))}.
\]
For $x$ outside of this range, the probability is $0$. We also define $\varepsilon = \exp(-p(b - a)/2)$. Then we can equivalently write:
\[
\Pr_{\mathbf x \sim \TrunExp(p, a, b)}(\mathbf x = x) = \frac{\exp(-p(x - a)) - \exp(-p(x - a + 1))}{1 - \varepsilon^2}.
\]
Equivalently, $\TrunExp(p, a, b)$ is the distribution that one gets when repeatedly sampling $x \sim a + \lfloor\Exp(p)\rfloor$ and keeping the first outcome that satisfies $x \in [a, b)$.
\end{definition}

The following properties will be useful later in the analysis:

\begin{remark}
\label{lem:truncated-exp-properties}
Let $\mathbf x \sim \TrunExp(p, a, b)$. For $a \le x < x + c \le b$, we have
\[
\Pr(\mathbf x \in [x, x + c)) = \frac{e^{-p(x - a)} - e^{-p(x + c - a)}}{1 - \varepsilon^2}
= e^{-p(x - a)}\cdot\frac{1 - e^{-pc}}{1 - \varepsilon^2}.
\]
Furthermore, if $c \ge (b - a) / 2$, we have
\[
\frac{1 - e^{-pc}}{1 - \varepsilon^2}
\ge
1 - e^{-pc}
\ge
1 - e^{-(b - a) / 2}
=
1-\varepsilon.
\]
and thus
\[
\Pr(\mathbf x \in [x, x + c)) \ge e^{-p(x - a)} \cdot (1 - \varepsilon) \ge \exp(-p(x - a) - 2\varepsilon),
\]
for sufficiently small $\varepsilon$.
Finally, for $a \le x < b - 1$, we have
\[
\Pr(\mathbf x = x + 1) = \frac{e^{-p(x + 1 - a)} - e^{-p(x + 2 - a)}}{1 - \varepsilon^2}
= 
e^{-p} \cdot \frac{e^{-p(x - a)} - e^{-p(x + 1 - a)}}{1 - \varepsilon^2}
= 
e^{-p} \cdot \Pr(\mathbf x = x).
\]
\end{remark}

For iteration $\ell$ of \cref{alg:bernhard-modified}, in the invoked \cref{alg:cutting-lemma}, we sample $r \sim \TrunExp(p_\ell, r_{\ell-1}, r_\ell)$ for $p_\ell = \frac{2\ln(2s_{\ell - 1}/\delta)}{r_\ell-r_{\ell - 1}}$.

In \cite{BFHL25}, it is shown that for iteration $\ell$ of \cref{alg:bernhard-modified}, each invocation of \cref{alg:cutting-lemma} cuts at most $2s_{\ell-1}$ balls.
For the purposes of our analysis, assume that exactly $2s_{\ell - 1}$ balls are cut, by padding with dummy balls of zero size.

The main idea of the analysis is simple: we have a lower bound on the probability that, assuming the algorithm has made some number of cuts without cutting our edge set, it will not cut the edge set later. We analyze this probability directly by looking at the current ball: with some probability the edge set survives the cut by this ball, a part of it ends up on the inside of the ball and it survives there, and the rest ends up on the outside of the ball and it survives there too. An important step in the analysis is to show that we do not need to analyze all the possible ways the edge set can be split, as the worst case happens either when all of the edge set ends up on the inside or on the outside of the cut, in a very specific way. After we show this, the rest of the proof boils down to showing a straightforward recurrence relation.

\begin{definition}
For all $i \in \{0, \ldots, 2L\}$, $0 \le j \le 2s_{\lceil L-i/2\rceil - 1}$, $m \ge 1$, $d \ge 0$, we define the quantity $\suc_{i, j}(m, d)$ as follows:

Suppose \cref{alg:bernhard-modified} is run on any graph $G$ with at most $m$ edges.
\cref{alg:cutting-lemma} has been invoked $i$ times\footnote{We count the invocations on $G$ and $\mathrm{rev}(G)$ in one iteration as two invocations.}, and in the $(i+1)$-th invocation, we have so far cut $j$ balls.
Let $\Gamma$ be any set of edges with total length $\le d$ that have not yet been cut or gone into recursion.
Conditioned on this, let $p$ be the probability that we will not cut $\Gamma$ in the rest of the algorithm. Then $\suc_{i, j}(m, d)$ is defined as the minimum such $p$ over all $G$, $\Gamma$ and possible runs of the algorithm that have led to this state.

Furthermore, define
\begin{align*}
Q_{i, j}(m, d) &=\exp\left(-\frac{d}{D} \cdot 640L \log m - 3L \log_{4/3} m \cdot \delta-\left(2L - i - \frac{j}{2s_{\ell - 1}}\right)\delta\right),
\end{align*}
where $\ell = \lceil L - i/2\rceil$. As a special case, $Q_{i, j}(m, 0) = 1$.
\end{definition}

\begin{lemma}
\label{lem:exp-cutting-induction}
For any $i,j,m$, and $d \le \frac{D}{8L}$, we have $\suc_{i, j}(m, d) \ge Q_{i,j}(m, d)$.
\end{lemma}

In the rest of this section, we prove \cref{lem:exp-cutting-induction}, which will be used for the proof of \cref{thm:dLDDgeneral}, by induction with increasing $m, d$ and decreasing $i, j$. For any $i$, we have $\ell = \lceil L - i/2 \rceil$, which is the value of $\ell$ used by the algorithm for the $(i+1)$-th invocation of the cutting lemma.

Let us start with the base cases: the claim holds trivially for $m \le 1$ or $d = 0$ or $(i, j) = (2L, 0)$, since then $\suc_{i, j}(m, d)$ is in fact 1, as the algorithm terminates without any cut.
Moreover, for $j = 2s_{\ell - 1}$, we have $\suc_{i,j}(m, d) = \suc_{i + 1, 0}(m, d)$ and $Q_{i,j}(m, d) = Q_{i + 1, 0}(m, d)$ and we can thus rely on the inductive hypothesis.

Now, let us have $m > 1$, $d > 0$, $i < 2L$, $j < 2s_{\ell - 1}$. Without loss of generality, we only consider the case of cutting out-balls. Fix the edge set $\Gamma$ of total length $d$ that we want to preserve. Fix $v$ the center of the $(j+1)$-th ball. We sample the radius $r \in [r_{\ell-1}, r_\ell)$. The set $\Gamma$ determines which values of $r$ are ``good'' or ``bad'', that is, if we sample this particular $r$ and make a cut, whether the set $\Gamma$ will survive or not. This can be formalized as follows: let $\ved$ be a vector indexed by $r \in [r_{\ell - 1}, r_\ell)$ such that $d_r$ counts the total number of edges of $\Gamma$ that would be cut by $\Bout(v, r)$, that is, edges $e = (x, y)$ with $\dist_G(v, x) \le r < \dist_G(v, y)$. Note that a value of $r$ is good if and only if $d_r = 0$, i.e., $r \notin \supp(\ved)$. We have $\|\ved\|_1 = d$ and $|\supp(\ved)| \le d$. The probability that $\Gamma$ survives is equal to
\[
\sum_{r \notin \supp(\ved)} \Pr(\mathbf r = r \land \text{$\Gamma$ survives the rest of the algorithm}).
\]
Let us express one term of this sum. Assuming the sampled radius is $r \notin \supp(\ved)$, let $B = \Bout(v, r)$ be the ball being cut. The set $\Gamma$ gets split into two: $\Gamma \cap B$, of total length $d_{< r} = \sum_{r' < r} d_{r'}$, ends up inside $B$, while $\Gamma \setminus B$ of total length $d_{>r} = d - d_{<r}$ ends up on the outside. $\Gamma$ will survive the rest of the algorithm if and only if $\Gamma \cap B$ survives the recursion and $\Gamma \setminus B$ survives the rest of the algorithm. It is shown in \cite{BFHL25} that the recursive call is on a graph with at most $\frac32 \cdot \frac{m}{s_\ell}$ edges. Thus, due to the independence of these two events, they occur with probability at least
\[
\suc_{0, 0}\left(\frac{3m}{2s_\ell}, d_{<r}\right) \cdot \suc_{i,j+1}(m, d_{>r})
\]
Summing over $r$, we know that the probability of $\Gamma$ surviving the rest of the algorithm is at least
\[
\sum_{r \notin \supp(\ved)} 
\Pr(r) \cdot
\suc_{0, 0}\left(\frac{3m}{2s_\ell}, d_{<r}\right) \cdot \suc_{i,j+1}(m, d_{>r})
.
\]
Fix $i, j, m, d$ and denote the above expression by $\ey(\ved)$.
Now necessarily $\suc_{i, j}(m, d) \ge \min_\ved \ey(\ved)$. To see this, take the minimizer of $\suc_{i,j}(m, d)$: that is, the $G$, $\Gamma$ and the set of cuts made by the algorithm so far for which the probability of $\Gamma$ surviving the rest of the algorithm is exactly $\suc_{i, j}(m, d)$. Then take the associated $\ved$ and get that $\suc_{i, j}(m, d) \ge \ey(\ved) \ge \min_\ved \ey(\ved)$ by the above analysis.

Now define $\ex$ as $\ey$ with $P$ replaced by $Q$:
\[
\ex(\ved) = \sum_{r \notin \supp(\ved)} 
\Pr(r) \cdot
Q_{0, 0}\left(\frac{3m}{2s_\ell}, d_{<r}\right) \cdot Q_{i,j+1}(m, d_{>r})
.
\]

We invoke the inductive hypothesis on the inner terms and get $\ey(\ved) \ge \ex(\ved)$.
The proof of \cref{lem:exp-cutting-induction} will be completed if we prove that $\ex(\ved) \ge Q_{i,j}(m, d)$ no matter what $\ved$ is, since this proves $\suc_{i, j}(m, d) \ge \min_\ved \ey(\ved) \ge \min_\ved \ex(\ved) \ge Q_{i, j}(m, d)$ as needed.
The rest of the proof will be split in two parts: first we prove that we can without loss of generality assume that $\ved$ is of one of two special vectors, and then we prove the inequality for both of them.

\paragraph{Finding the minimizer of $\ex(\ved)$.}
Define $\ved^\leftarrow$ as the 0-1 vector with $\supp(\ved^\leftarrow) = \{r_{\ell - 1}, \ldots, r_{\ell - 1} + d - 1\}$. Similarly, define $\ved^\rightarrow$ as the 0-1 vector with $\supp(\ved^\rightarrow) = \{r_\ell - d, \ldots, r_\ell - 1\}$. We show that one of $\ved^\leftarrow,\ved^\rightarrow$ is always a minimizer of $\ex(\ved)$. We show this in two steps: first we show that a non-0-1 vector is never a minimizer, then we show that any 0-1 vector can be gradually transformed into either $\ved^\leftarrow$ or $\ved^\rightarrow$ by a series of local changes that do not increase the value of $\ex(\ved)$.

First assume $\ved$ is not a 0-1 vector. Let $i$ be any index such that $d_i > 1$ and let $j$ be the closest index to $i$ such that $d_j = 0$ (so all coordinates between $d_i,d_j$ are nonzero). Construct a new $\ved'$ such that $d'_i = d_i - 1$, $d'_j = 1$ and $d'_k = d_k$ for all other $k$. Note that $\supp(\ved') = \supp(\ved) \setminus \{j\}$ and that for all $r \in \supp(\ved')$, we have $d_{<r} = d'_{<r}$ and $d_{>r} = d'_{>r}$. In other words, the terms present in both $\ex(\ved)$ and $\ex(\ved')$ are the same, except that $\ex(\ved')$ lacks the term for $r = j$. Thus, $\ex(\ved) \ge \ex(\ved')$.

Now we know that the minimizing $\ved$ is a 0-1 vector. Consider two 0-1 vectors $\ved, \ved'$ that are nearly identical, except that for some $r$, we have $d_r = d'_{r + 1} = 0$ and $d_{r + 1} = d'_r = 1$. We can write:
\begin{align*}
\ex(\ved) - \ex(\ved')
=&
\Pr(r) \cdot
Q_{0, 0}\left(\frac{3m}{2s_\ell}, d_{<r}\right) \cdot Q_{i,j+1}(m, d_{>r})
\\&
-
\Pr(r + 1) \cdot
Q_{0, 0}\left(\frac{3m}{2s_\ell}, d_{<r}+1 \right) \cdot Q_{i,j+1}(m, d_{>r} - 1),
\end{align*}
where we used $d'_{<r + 1} = d_{<r} + 1$ and $d'_{>r+1} = d_{>r} - 1$, and the fact that all the other terms cancel out.
Now let
\[
\alpha =
\frac{
\Pr(r) \cdot
Q_{0, 0}\left(\frac{3m}{2s_\ell}, d_{<r}\right) \cdot Q_{i,j+1}(m, d_{>r})
}{
\Pr(r + 1) \cdot
Q_{0, 0}\left(\frac{3m}{2s_\ell}, d_{<r}+1 \right) \cdot Q_{i,j+1}(m, d_{>r} - 1),
}.
\]
Clearly, $\ex(\ved) > \ex(\ved')$ if and only if $\alpha > 1$.
After we expand all terms and cancel, we are left with
\[
\alpha = \exp(p_\ell) \cdot \frac{\exp(-\frac{1}D \cdot 640L\log m)}{\exp\left(-\frac1D \cdot 640L\log \left(\frac{3m}{2s_\ell}\right)\right)}.
\]
where we use \cref{lem:truncated-exp-properties} to show that
\[
\frac{\Pr(\textbf r = r)}{\Pr(\textbf r = r + 1)}
=
\exp(p_\ell).
\]

Note that the value of $\alpha$ does not depend on $r$, $\ved$, or $\ved'$. Suppose $\alpha \ge 1$. Then $\ex(\ved) \ge \ex(\ved')$. We can now repeat the same process by finding a subsequence $(0, 1)$ in $\ved'$ and changing it to $(1, 0)$ to construct $\ved''$. Since $\alpha$ does not depend on $r$, $\ved'$ or $\ved''$, the same analysis concludes that $\ex(\ved') \ge \ex(\ved'')$. We can continue this process until there is no $(0,1)$ subsequence in the vector, that is, until we end up with $\ved^\leftarrow$. What we have shown is that assuming $\alpha \ge 1$, we have $\ex(\ved) \ge \ex(\ved^\leftarrow)$, and this holds for all $\ved$. Analogously, if instead $\alpha < 1$, we can do the same operation in reverse, changing $(1, 0)$-subsequences to $(0, 1)$ until there are no more $(1, 0)$-subsequences in the vector. This way, we show that if $\alpha < 1$, then $\ex(\ved) > \ex(\ved^\rightarrow)$ for all $\ved$.

\paragraph{Bounding $Q_{i,j}(m, d)$.}
All that remains now is to show that $\ex(\ved^\leftarrow) \ge Q_{i, j}(m, d)$ and $\ex(\ved^\rightarrow) \ge Q_{i, j}(m, d)$. We start with the former. Recall that $\supp(\ved^\leftarrow) = \{r_{\ell - 1}, \ldots, r_{\ell - 1} + d - 1\}$. For $r \notin \supp(\ved^\leftarrow)$, we have $d^\leftarrow_{<r} = d$, $d^\leftarrow_{>r} = 0$. The whole expression thus simplifies to:
\begin{align*}
\ex(\ved^\leftarrow)
&=
\sum_{r \ge r_{\ell - 1} + d}
\Pr(r) \cdot
Q_{0, 0}\left(\frac{3m}{2s_\ell}, d\right) \cdot \underbrace{Q_{i,j+1}(m, 0)}_{=1}
\\&=
\left(\sum_{r \ge r_{\ell - 1} + d}
\Pr(r)\right)
\cdot
Q_{0, 0}\left(\frac{3m}{2s_\ell}, d\right)
\end{align*}

Recall that $d \le (r_\ell-r_{\ell - 1})/2$ by our assumption and $\varepsilon = \delta / 2s_{\ell - 1} = o(1)$. Thus, we can use \cref{lem:truncated-exp-properties} to bound:

\begin{align*}
\sum_{r \ge r_{\ell - 1} + d}
\Pr(r)
&=
\Pr(\mathbf r \in [r_{\ell - 1} + d, r_\ell))
\ge
\exp(-p_\ell d - 2\varepsilon)
\ge
\exp(-p_\ell d - L\delta),
\end{align*}
where in the last step, we used a very crude bound on $2\varepsilon = \delta / s_{\ell - 1} \le L\delta$.

\def\m{\frac{3m}{2s_\ell}}
We also expand:
\begin{align*}
Q_{0, 0}\left(\frac{3m}{2s_\ell}, d\right)
&=
\exp\left(-\frac{d}{D} \cdot 640L \log \m - 3L\log_{4/3} \left(\m\right)\cdot \delta - 2L \delta \right)
\\&\ge
\exp\left(-\frac{d}{D} \cdot 640L \log \left(\frac{m}{2^{2^{L-\ell-1}}}\right) -3L\log_{4/3} \left(\frac34 m\right)\cdot\delta - 2L\delta\right),
\\&=
\exp\left(-\frac{d}{D} \cdot 640L \left(\log m - 2^{L-\ell-1}\right)  -3L(\log_{4/3} m - 1)\delta - 2L\delta \right),
\end{align*}
where we use $\frac{3m}{2s_{\ell}} \le \frac34 m$ and $\frac{3m}{2s_\ell} \le m / 2^{2^{L - \ell - 1}}$.

Before we put everything together, let us also bound $p_\ell d$ identically to \cite{BFHL25}:
\begin{align*}
p_\ell d
    &= \frac{2d \ln(2s_{\ell-1} / \delta)}{r_\ell - r_{\ell-1}} \\
    &\leq \frac{2d \ln(2^{2^{L - \ell + 1}} \cdot (\log m)^{10})}{\frac{D}{2^{L-\ell+3}} + \frac{D}{4L}} \\
    &\leq \frac{2d}{D} \cdot \frac{2^{L - \ell + 1} + 10L}{\max\set*{\frac{1}{2^{L-\ell+3}}, \frac{1}{4L}}} \\
    &\leq \frac{2d}{D} \cdot (2^{L-\ell+3} + 10L) \cdot \min\set*{2^{L-\ell+3}, 10L} \\
    &\leq \frac{d}{D} \cdot 2^{L-\ell} \cdot 320L.
\end{align*}
In the last step we used that $(a + b) \min\set{a, b} \leq 2 \max\set{a, b} \min\set{a, b} = 2ab$.

Now we can write:
\begin{align*}
\log \ex(\ved^\leftarrow)
&\ge
-p_\ell d - L \delta + \log Q_{0, 0} \left(\m, d\right)
\\&\ge
-\frac dD \cdot 2^{L-\ell}\cdot 320L - L\delta - \frac{d}{D} \cdot 640L \left(\log m - 2^{L-\ell-1}\right)  -3L(\log_{4/3} m - 1)\delta - 2L\delta 
\\&=
-\frac dD \cdot 640L \log m - 3L\log_{4/3} m \cdot \delta
\\&= \log Q_{2L,0}(m, d) \ge \log Q_{i, j}(m, d).
\end{align*}

Similarly, we can calculate $\ex(\ved^\rightarrow)$. Recall that $\supp(\ved^\rightarrow) = \{r_\ell - d, \ldots, r_\ell - 1\}$.
\begin{align*}
\ex(\ved^\rightarrow)
&=
\sum_{r < r_\ell - d}
\Pr(r) \cdot
\underbrace{Q_{0, 0}\left(\frac{3m}{2s_\ell}, 0\right)}_{=1} \cdot Q_{i,j+1}(m, d)
\\&=
\left(\sum_{r < r_\ell - d}
\Pr(r)\right)
\cdot
Q_{i,j+1}(m, d)
\end{align*}

Now we again start by bounding the first term using \cref{lem:truncated-exp-properties}:
\begin{align*}
\Pr(\mathbf r \in [r_{\ell - 1}, r_\ell - d))
\ge
\exp(- \varepsilon).
\end{align*}

Recall that $\varepsilon = \delta / (2s_{\ell - 1})$. We can directly write
\begin{align*}
\log \ex(\ved^\rightarrow) &\ge 
-\varepsilon + \log Q_{i, j + 1}(m, d)
\\&\ge
-\frac\delta{2s_{\ell - 1}}
-\frac{d}{D} \cdot 640L \log m - 3L \log_{4/3} m \cdot \delta-\left(2L - i - \frac{j + 1}{2s_{\ell - 1}}\right)\delta
\\&=
-\frac{d}{D} \cdot 640L \log m - 3L \log_{4/3} m \cdot \delta-\left(2L - i - \frac j{2s_{\ell - 1}}\right)\delta
\\&=\log Q_{i,j}(m, d).\qedhere
\end{align*}

Finally, we are ready to prove \cref{thm:geom-modified}.

\begin{proof}[Proof of \cref{thm:geom-modified}]
As discussed previously, all invocations of Cohen's algorithm and step 2.1 of \cref{alg:cutting-lemma} can be implemented to succeed with high probability $1-1/\poly(m)$.
Together with \cref{lem:exp-cutting-induction} and by union bound, any subset $\Gamma$ of edges with total length $d_\Gamma$ is not cut with probability at least $Q_{0,0}(m,d_\Gamma)-1/\poly(m)$.
The extra $1/\poly\log(m)$ term in the exponent is subsumed if $d_\Gamma/D \cdot \log m \log\log m$ is at least some small constant.
Otherwise, the desired bound already follows from \cite{BFHL25}.
The extra $1/\poly(m)$ loss can always be ignored as we only consider $d_\Gamma \le \epsilon D/\log\log m$ for sufficiently small $\epsilon$.
Note that guarantees for edges and paths are direct corollaries of guarantees for subsets of edges.

To see the independence property, we first note that distance between any two vertices in the remaining graph can only increase after each cutting step, and that all balls being cut have radius at most $D$.
Then observe that in each cutting step, independent randomness is used
and either $\Gamma$ or $X \cup Y$
can never be cut due to the assumption.
Thus, the same probability bound for $\Gamma$ being cut also holds if it is conditioned on $X,Y$.

The theorem follows as our modification does not affect the diameter property and the running time proved in \cite{BFHL25}.
\end{proof}

\section{Strong Bounds for Modified \cite{Li25}'s Algorithm}
\label{sec:jason-modified}

In this section, we modify \cite{Li25}'s algorithm to prove the following result, which is a more formal version of \cref{thm:dLDDseparated}, with probability guarantees for edges at the same time.
The algorithm is presented in \cref{sec:jason-modified-alg}, and \cref{sec:jason-modified-proof} constitutes the proof.

\begin{theorem}
\label{thm:jason-modified}
Let $\epsilon > 0$ be a sufficiently small constant.
There is a randomized algorithm that, given any directed graph $G=(V,E)$ with $n$ vertices and $m$ edges, diameter parameter $D=\poly(n)$, and separation parameter $d$, returns an ordered (full) clustering $\sigma=(C_1,\ldots,C_k)$ of $V$ and boolean vertex marks $\mk(\cdot)$, satisfying all of the following.
\begin{description}[labelindent=.5cm]
    \item[Clustering probability:] Each vertex $v$ is \emph{clustered}, i.e., $\mk(v)=0$, with probability at least $\exp(-d/D \cdot \O(\log n \log\log n))$ if $d \le \epsilon D / \log\log n$.
    \item[Edge cutting probability:] Each edge $e=(u,v)$ is \emph{cut}, i.e., $u \in C_i$ and $v \in C_j$ for $i > j$, with probability at most $1-\exp(-d_e/D \cdot \O(\log n \log\log n))$ if $d_e \le \epsilon D / \log\log n$.
    \item[Separation property:] For any \emph{clustered} vertices $u \in C_i$ and $v \in C_j$ with $i > j$, it holds that $\dist_G(u,v) > d$.
    \item[Diameter property:] Each $C_i$ is strongly connected and has diameter at most $D$.
    \item[Running time:] The algorithm runs in $\O((m+n\log\log n)\log^2 n)$ time with high probability on graphs with polynomially bounded, integral edge lengths.
\end{description}
\end{theorem}

\subsection{Overview of \cite{Li25}'s Algorithm and Our Modification}
\label{sec:jason-modified-alg}

For a union $\Bout = \bigcup_{i=1}^k \Bout_G(t_i,r)$ of out-balls, define $\dist_{\Bout}(v) = \min_{i \in [k]} \dist_G(t_i,v)$.
Define $\dist_{\Bin}(\cdot)$ for a union $\Bin$ of in-balls similarly.
Note that $B^\pm = \set{v \mid \dist_{B^\pm}(v) \le r}$.

The algorithm is shown in \cref{alg:jason-modified}.

\begin{remark}
\label{rmk:alg}
\cref{alg:jason-modified} is essentially \cite{Li25}'s algorithm with following modifications.
\begin{enumerate}
    \item The separation parameter $d$ is used only for clustering, while having nothing to do with cutting.
        Intuitively, each time a cut is executed, vertices within distance $d$ to the boundary are marked and excluded from clustering.
    \item Exactly one of iterations $i(\pm)$ is executed in an alternating way, while not affecting correctness.
        Consequently, heavy vertex elimination is not needed in cases not having both in-heavy and out-heavy vertices.
    \item In case (b) of having both in-heavy and out-heavy vertices, following changes are made to simplify the algorithm as well as analysis.
        \begin{enumerate}
            \item Both $\Bin_G$ and $\Bout_G$ are computed, which are disjoint by the condition of case (b), so the recursive instance has at most $m/2$ edges.
            \item Cutting $\Bout$ eliminates in-heavy vertices, so step 6 can be executed.
                Similarly, cutting $\Bin$ eliminates out-heavy vertices, so step 7 can be executed.
        \end{enumerate}
\end{enumerate}
\end{remark}

\begin{algorithm}[p]
\caption{Modified \cite{Li25}'s Algorithm}
\label{alg:jason-modified}
\begin{description}
    \item[Input:] graph $G=(V,E)$ with $n$ vertices and $m$ edges, diameter $D$, separation $d$.
    \item[Output:] ordered clustering $\sigma$, boolean marks $\mk(\cdot)$.
\end{description}
\begin{enumerate}
    \item Let $L=\log\log m$, and $a_i = D/8 - \sum_{j=1}^i D/16 \cdot \max(1/L,2^{-i})$ for $i \in [0,L]$.
    \item Set $U \gets V$, $\sigma^+,\sigma^- \gets \bot$, and $\mk(v) \gets 0$ for $v \in V$.
    \item If $m \le 1$, return the singleton clustering $\sigma$ induced by a topological order of $G$, and $\mk(\cdot)$.
    \item Label each vertex as either in-heavy or in-light such that $|E[\Bin_G(v,D/8]| \ge m/2$ for each in-heavy vertex $v$, and that $|E[\Bin_G(v,D/8]| \le 3m/4$ for each in-light vertex $v$.
    \item Label each vertex as either out-heavy or out-light such that $|E[\Bout_G(v,D/8]| \ge m/2$ for each out-heavy vertex $v$, and that $|E[\Bout_G(v,D/8]| \le 3m/4$ for each out-light vertex $v$.
    \item If there is no in-heavy vertex,
        \begin{enumerate}
            \item For $i=1,\ldots,L$, execute iteration $i(-)$ for odd $i$, and iteration $i(+)$ for even $i$.
                \begin{description}
                    \item[Iteration $i(-)$:] \mbox{}
                        \begin{enumerate}
                            \item Sample $r^-_i \in (a_i,a_{i-1}]$, and sample each vertex $v \in U$ with probability $\min(2\deg(v)/m \cdot 2^{2^i}\ln (mD),1)$.
                                Let $S^-_i$ be the set of sampled vertices.
                            \item For $v \in S^-_i$ in random order,
                                \begin{enumerate}[label=(\Alph*)]
                                    \item Set $\mk(v) \gets 1$ for $v \in (\Bin_G(v,r^-_i+d) \setminus \Bin_G(v,r^-_i-d)) \cap U$.
                                    \item Recurse on $G[B=\Bin_G(v,r^-_i) \cap U]$.
                                        Let $\sigma',\mk'(\cdot)$ be returned by the recursion.
                                    \item Set $U \gets U \setminus B$, $\sigma^- \gets \sigma^- \circ \sigma'$, and $\mk(v) \gets \mk(v) \lor \mk'(v)$ for $v \in B$.
                                \end{enumerate}
                        \end{enumerate}
                    \item[Iteration $i(+)$:] \mbox{}
                        \begin{enumerate}
                            \item Sample $r^+_i \in (a_i,a_{i-1}]$, and sample each vertex $v \in U$ with probability $\min(2\deg(v)/m \cdot 2^{2^i}\ln (mD),1)$.
                                Let $S^+_i$ be the set of sampled vertices.
                            \item For $v \in S^+_i$ in random order,
                                \begin{enumerate}[label=(\Alph*)]
                                    \item Set $\mk(v) \gets 1$ for $v \in (\Bout_G(v,r^+_i+d) \setminus \Bout_G(v,r^+_i-d)) \cap U$.
                                    \item Recurse on $G[B=\Bout_G(v,r^+_i) \cap U]$.
                                        Let $\sigma',\mk'(\cdot)$ be returned by the recursion.
                                    \item Set $U \gets U \setminus B$, $\sigma^+ \gets \sigma' \circ \sigma^+$, and $\mk(v) \gets \mk(v) \lor \mk'(v)$ for $v \in B$.
                                \end{enumerate}
                        \end{enumerate}
                \end{description}
        \end{enumerate}
    \item Else if there is no out-heavy vertex,
        \begin{enumerate}
            \item For $i=1,\ldots,L$, execute iteration $i(+)$ for odd $i$, and iteration $i(-)$ for even $i$.
        \end{enumerate}
\end{enumerate}
\end{algorithm}

\begin{algorithm}[p]
\ContinuedFloat
\caption{Modified \cite{Li25}'s Algorithm (Continued)}
\begin{enumerate}[start=8]
    \item Else if there are both in-heavy and out-heavy vertices,
        \begin{enumerate}
            \item If there is an in-heavy vertex $s$ and an out-heavy vertex $t$ such that $\dist_G(s,t) \le D/4$,
                \begin{enumerate}
                    \item Sample $r \in (D/8,D/4]$.
                    \item Set $\mk(v) \gets 1$ for $v \in \Bin_G(s,r+d) \setminus \Bin_G(s,r-d)$.
                    \item Recurse on $G[B=V \setminus \Bin_G(s,r)]$.
                        Let $\sigma',\mk'(\cdot)$ be returned by the recursion.
                    \item Set $\sigma^+ \gets \sigma'$, and $\mk(v) \gets \mk(v) \lor \mk'(v)$ for $v \in B$.
                    \item Set $\mk(v) \gets 1$ for $v \in (\Bout_G(t,r+d) \setminus \Bout_G(t,r-d)) \cap \Bin_G(s,r)$.
                    \item Recurse on $G[B=\Bin_G(s,r) \setminus \Bout_G(t,r)]$.
                        Let $\sigma',\mk'(\cdot)$ be returned by the recursion.
                    \item Set $\sigma^- \gets \sigma' \circ (\Bin_G(s,r) \cap \Bout_G(t,r))$, and $\mk(v) \gets \mk(v) \lor \mk'(v)$ for $v \in B$.
                \end{enumerate}
            \item Else if $\dist_G(s,t) > D/4$ for any in-heavy vertex $s$ and out-heavy vertex $t$,
                \begin{enumerate}
                    \item Sample $r \in (D/16,D/8]$
                    \item Let $\Bin$ be the union of $\Bin_G(t,r)$ over all out-heavy vertices $t$, and $\Bout$ the union of $\Bout_G(s,r)$ over all in-heavy vertices $s$.
                    \item If $|E[\Bin]| \ge |E[\Bout]|$,
                        \begin{enumerate}[label=(\Alph*)]
                            \item Set $\mk(v) \gets 1$ for $v \in V$ if $\dist_{\Bout}(v) \in (r-d,r+d]$.
                            \item Recurse on $G[\Bout]$.
                                Let $\sigma',\mk'(\cdot)$ be returned by the recursion.
                            \item Set $U \gets V \setminus \Bout$, $\sigma^+ \gets \sigma'$, and $\mk(v) \gets \mk(v) \lor \mk'(v)$ for $v \in \Bout$.
                            \item Execute step 6.
                        \end{enumerate}
                    \item Else if $|E[\Bout]| > |E[\Bin]|$,
                        \begin{enumerate}[label=(\Alph*)]
                            \item Set $\mk(v) \gets 1$ for $v \in V$ if $\dist_{\Bin}(v) \in (r-d,r+d]$ .
                            \item Recurse on $G[\Bin]$.
                                Let $\sigma',\mk'(\cdot)$ be returned by the recursion.
                            \item Set $U \gets V \setminus \Bin$, $\sigma^- \gets \sigma'$, and $\mk(v) \gets \mk(v) \lor \mk'(v)$ for $v \in \Bin$.
                            \item Execute step 7.
                        \end{enumerate}
                \end{enumerate}
        \end{enumerate}
    \item Return $\sigma^- \circ \sigma^+$ and $\mk(\cdot)$.
\end{enumerate}
\end{algorithm}

\subsection{Analysis of Modified \cite{Li25}'s Algorithm}
\label{sec:jason-modified-proof}

We show that \cref{alg:jason-modified} satisfies all of the requirements of \cref{thm:jason-modified}.
In the rest of this section, assume $\epsilon > 0$ is a sufficiently small constant and $D=\poly(m)$.

\paragraph{Clustering probability.}

Assume $d \le \epsilon D/\log\log(mD)$.
We use the following lemma of \cite{Li25}.

\begin{lemma}[\cite{Li25}'s Lemma 5]
\label{lem:prob-sample}
For $i \in [L]$, with probability at least $1-(mD)^{-3}$, all remaining vertices $v \in U$ satisfy $\vol_G(B^\pm_G(v,a_i) \cap U) \le 2m/2^{2^i}$ simultaneously after iteration $i(\mp)$.
\end{lemma}

We also use the following technical lemma.
Its proof is deferred to \cref{append:proof-jason}.

\begin{lemma}
\label{lem:avg-ratio-prefix-d}
For integers $r,d \ge 1$ such that $r \ge 2d$, and reals $k,k_1,\ldots,k_r \ge 0$ such that $k_1 \ge 1$ and $k_1+\ldots+k_r=k$, it holds that
\[
\frac{1}{r} \cdot \sum_{i=1}^r \frac{\sum_{j=\max(i-d+1,1)}^i k_i}{\sum_{j=1}^i k_i} \le 1-\left(1-\frac{1}{\lfloor r/d \rfloor}\right) \cdot k^{-\frac{1}{\lfloor r/d \rfloor - 1}}.
\]
\end{lemma}

At a high level, the proof of clustering probability follows a similar approach to \cite{Li25} with two crucial differences.
\begin{itemize}
    \item A much more delicate analysis is needed to derive an exponential bound on the probability of not marking a vertex in each iteration.
    The argument still relies on the random order of sampled vertices, while additionally exploiting the important observation that only $a_{i-1}-a_i$ possibilities of the distance may be relevant.
    Even with $k \gg a_{i-1}-a_i$ sampled vertices, we can obtain a much tighter bound than the one in \cite{Li25} that depends linearly on $\ln k$.
    We elaborate on this in \cref{lem:prob-mark}.
    \item To maintain such an exponential bound throughout the algorithm, independence between iterations and recursions must be carefully examined, as will be seen in \cref{lem:prob-cluster}.
\end{itemize}

\begin{lemma}
\label{lem:prob-mark}
For $i \in [L]$ and vertex $v \in U$ unmarked before iteration $i(\pm)$, $v$ remains unmarked after iteration $i(\pm)$ with probability at least $\exp(-d/D \cdot \O(2^i\log\log(mD)))$.
\end{lemma}

\begin{proof}
We prove for iteration $i(+)$.
Iteration $i(-)$ is symmetric.
Let $k_j$ be the number of sampled vertices $t \in S^+_i$ such that $\dist_B(t,v)=j$, and $k=\sum k_j$.
$v$ is marked by radius $r^+_i$ if and only if
\begin{itemize}
    \item Radius $r^+_i$ is sampled with probability $1/(a_{i-1}-a_i)$.
    \item Some $t \in S^+_i$ with $\dist_G(t,v) \in (r^+_i-d,r^+_i+d]$ appears first among all $t' \in S^+_i$ with $\dist_G(t',v) \le r^+_i+d$.
        This happens with probability\footnote{$\frac{0}{0}$ is treated as $0$.}
        \[
        \frac{\sum_{j=\max(r^+_i-d+1,0)}^{r^+_i+d} k_j}{\sum_{j=0}^{r^+_i+d} k_j}.
        \]
\end{itemize}
Altogether, $v$ is marked with probability
\[
\frac{1}{a_{i-1}-a_i} \cdot \sum_{r^+_i=a_i+1}^{a_{i-1}} \frac{\sum_{j=\max(r^+_i-d+1,0)}^{r^+_i+d} k_j}{\sum_{j=0}^{r^+_i+d} k_j}.
\]
Without loss of generality, we make the following assumptions without decreasing the probability.
\begin{description}[labelindent=.5cm]
    \item[$\bf k_0 = \cdots = k_{a_i+d}=0$:]
        Otherwise, we set $k_{a_i+d+1}$ to take the value of $\sum_{j=0}^{a_i+d+1} k_j$.
        Indeed, for each term of the above summation, its denominator is unchanged while its numerator may only increase.
        So the sum never decreases.
    \item[$\bf k_j=0$ for $\bf j>a_{i-1}+d$:]
        Otherwise, we set $k_{a_{i-1}+d}$ to take the value of $\sum_{j \ge a_{i-1}+d} k_j$.
        It may only increase the last term of the summation.
    \item[$\bf k_{a_i+d+1} \ge 1$:]
        Otherwise, we set $k_{a_i+d+1}=1$ and decrease $k_j$ by $1$ for the smallest $j$ such that $k_j \ge 1$.
        This results in an increase from $0$ to $1$ for $\min(j-(a_i+d+1),2d)$ terms starting from $a_i+d+1$ and a decrease by at most $1$ for $\min(j-(a_i+d+1),2d)$ terms starting from $j$.
        The net effect is thus nonnegative.
\end{description}
With all these assumptions, the probability of $v$ being marked is at most
\[
\frac{1}{a_{i-1}-a_i} \cdot \sum_{r^+_i=a_i+1}^{a_{i-1}} \frac{\sum_{j=\max(r^+_i-d+1,a_i+d+1)}^{r^+_i+d} k_j}{\sum_{j=a_i+d+1}^{r^+_i+d} k_j},
\]
with $k=\sum_{j=a_i+d+1}^{a_{i-1}+d} k_j$.
Applying \cref{lem:avg-ratio-prefix-d} with parameters $a_{i-1}-a_i$ and $2d$, we get that the probability of $v$ not being marked is at least $\exp(-\O(d/(a_{i-1}-a_i) \cdot \ln k)$ because $1-1/z \ge \exp(-1/(z-1))$.
Due to convexity, the same calculation as in \cite{Li25}'s Lemma 6, taking expectation over $k$, concludes the proof.
\end{proof}

\begin{lemma}
\label{lem:prob-cluster}
Let $p(m)$ be such that for any vertex of any graph with at most $m$ edges, it is never marked during \cref{alg:jason-modified} with probability at least $p(m)$.
It holds that $p(m) \ge \exp(-d/D \cdot \O(\log m\log\log(mD)))$.
\end{lemma}

\begin{proof}
We claim that $p(\cdot)$ satisfies the recursion
\[
p(m) \ge \min_{\alpha \in [c_1,m]} \exp(-\frac{c_2 d}{D} \cdot \log\alpha \log\log(mD)) \cdot p\left(\frac{m}{\alpha}\right) - \frac{1}{(mD)^{c_3}},
\]
for some constants $c_1,c_2 > 1$ and $c_3 \ge 2$, with the base case $p(1)=1$.
Assuming the claim for now, we can get
\[
p(m) \ge \exp\left(-\frac{c_4 d}{D} \cdot \log m \log\log(mD)\right),
\]
where $c_4 = c_2 + 1/\log c_1$, proving the lemma.
Indeed, note that RHS evaluates to $1$ for $m=1$, which is satisfied by \cref{alg:jason-modified} as no vertex is marked.
For $m > 1$, observe that by induction,
\begin{align*}
& \exp\left(-\frac{c_2 d}{D} \cdot \log\alpha \log\log(mD)\right) \cdot p\left(\frac{m}{\alpha}\right) - \frac{1}{(mD)^{c_3}}\\
& \hspace{.5cm} \ge \exp\left(-\frac{c_2 d}{D} \cdot \log\alpha \log\log(mD)\right) \cdot \exp\left(-\frac{c_4 d}{D} \cdot \log\frac{m}{\alpha} \log\log\frac{mD}{\alpha}\right) - \frac{1}{(mD)^{c_3}}\\
& \hspace{.5cm} \ge \exp\left(\frac{(c_4-c_2) d}{D} \cdot \log\alpha \log\log(mD) -\frac{c_4 d}{D} \cdot \log\alpha \log\log(mD) - \frac{c_4 d}{D} \cdot \log\frac{m}{\alpha} \log\log(mD)\right)\\
& \hspace{2cm} - \exp(-c_3 \cdot \ln(mD))\\
& \hspace{.5cm} \ge \exp\left(\frac{(c_4-c_2) d}{D} \cdot \log\alpha \log\log(mD) -\frac{c_4 d}{D} \cdot \log m \log\log(mD)\right) \cdot \biggl(1 - \exp\left(-\frac{c_3}{2} \cdot \ln(mD)\right)\biggr),
\end{align*}
because the exponent of the first factor is at least
\[
-c_4 \epsilon \log m \ge -c_3/2 \cdot \ln(mD),
\]
for $\epsilon \le c_3/(2c_4) \cdot \ln 2$.
Also note that
\begin{align*}
1 - \exp\left(-\frac{c_3}{2} \cdot \ln(mD)\right)
& = 1 - \frac{1}{(mD)^{c_3/2}}\\
& \ge \exp\left(-\frac{1}{(mD)^{c_3/2}-1}\right)\\
& \ge \exp\left(-\frac{1}{mD-1}\right)\\
& \ge \exp\left(-\frac{1}{D}\right),
\end{align*}
and that
\[
(c_4-c_2) \log\alpha \ge (c_4-c_2) \log c_1 = 1.
\]
Thus, all above together shows that the recursion is indeed satisfied.

It remains to prove the claim.
By union bound, with probability at least $1-2L(mD)^{-3} \ge 1-(mD)^{-2}$, all remaining vertices in $U$ satisfy the statement of \cref{lem:prob-sample} simultaneously in all $2L$ iterations.
Also, as noted in \cite{Li25}, vertices can be efficiently labeled as heavy or light with high probability by \cite{BringmannCF23}.
Its failure probability is subsumed by $(mD)^{-2}$ so long as $D=\poly(m)$.
We assume there is no sampling or labeling failure in the rest of the proof.
Let $c_0 > 0$ be such that the probability in \cref{lem:prob-mark} is at least $\exp(-c_0 d/D \cdot 2^i \log\log(mD))$.

In cases not having both in-heavy and out-heavy vertices, consider any unmarked vertex $v \in U$ at the beginning of iteration $i(\pm)$.
By \cref{lem:prob-mark}, it remains unmarked in the current iteration with probability at least $\exp(-c_0 d/D \cdot 2^i \log\log(mD))$.
Conditioned on this event, it goes into recursion with some probability $q$, and continues to the next iteration with probability $1-q$.

Since we are lower bounding the probability, it is without loss of generality to assume either $q=1$ or $q=0$.
If $q=1$, the recursive instance has at most $1.5m/2^{2^{i-1}}$ edges.
This is because for $i=1$, $G[B^\pm_G(v,r^\pm_1) \cap U]$ has at most $|E[B^\pm_G(v,r^\pm_1) \cap U]| \le |E[B^\pm_G(v,D/8) \cap U]| \le 3m/4$ edges, due to $r^\pm_1 \le a_0 = D/8$ and the assumption that there is no out-heavy or in-heavy vertex, respectively.
Meanwhile, for $i>1$, since \cref{alg:jason-modified} executes iterations in an alternating way, $G[B^\pm_G(v,r^\pm_i) \cap U]$ has at most $\vol_G(B^\pm_G(v,r^\pm_i) \cap U)/2 \le \vol_G(B^\pm_G(v,a_{i-1}) \cap U)/2 \le m/2^{2^{i-1}}$ edges after iteration $(i-1)(\mp)$, due to $r^\pm_i \le a_{i-1}$ and \cref{lem:prob-sample}.
Since all computation in recursion uses independent randomness, the probability is lowered by a multiplicative factor of at least $p(1.5m/2^{2^{i-1}})$.

If $q=0$, \cref{alg:jason-modified} will eventually go into recursion at some later iteration because each vertex in $U$ is sampled with probability $1$ in the last iteration.
Furthermore, since all computation in different iterations also use independent randomness, it implies
\begin{align*}
p(m)
& \ge \min_{i \in [L]} \left(\prod_{j=1}^i \exp\left(-\frac{c_0 d}{D} \cdot 2^j \log\log(mD)\right)\right) \cdot p\left(\frac{1.5m}{2^{2^{i-1}}}\right) - \frac{1}{(mD)^2}\\
& = \min_{i \in [L]} \exp\left(-\frac{c_0 d}{D} \cdot \sum_{j=1}^i 2^j \log\log(mD)\right) \cdot p\left(\frac{1.5m}{2^{2^{i-1}}}\right) - \frac{1}{(mD)^2}\\
& \ge \min_{i \in [L]} \exp\left(-\frac{c_0 d}{D} \cdot 2^{i+1} \log\log(mD)\right) \cdot p\left(\frac{1.5m}{2^{2^{i-1}}}\right) - \frac{1}{(mD)^2}
\end{align*}
By setting $c_1=4/3$, $c_2=\max(10 c_0,64)$, $c_3=2$, and $\alpha=2^{2^{i-1}}/1.5$, we get
\[
p(m) \ge \min_{\alpha \in [c_1,m]} \exp\left(-\frac{c_2 d}{D} \cdot \log\alpha \log\log(mD)\right) \cdot p\left(\frac{m}{\alpha}\right) - \frac{1}{(mD)^{c_3}},
\]
as claimed, because the minimization is relaxed to $\alpha \in [c_1,m]$.

If there are both in-heavy and out-heavy vertices, we verify that the inequality also holds.
In case (a), any vertex is marked with probability at most $2 \cdot 2d/(D/4-D/8) = 32d/D$ before possibly going into recursion.
In other words, it is not marked with probability at least $1 - 32d/D \ge \exp(-64d/D)$ for sufficiently small $\epsilon$.
The recursive instance has at most $m/2$ edges as a heavy ball is excluded.
So we have
\[
p(m) \ge \exp\left(-\frac{c_2 d}{D} \cdot \log\alpha \log\log(mD)\right) \cdot p\left(\frac{m}{\alpha}\right) - \frac{1}{(mD)^{c_3}},
\]
for $\alpha=2$, as desired.

In case (b), any vertex is marked with probability at most $2d/(D/8-D/16)=32d/D$.
Equivalently, it is not marked with probability at least $1-32d/D \ge \exp(-64d/D)$ for sufficiently small $\epsilon$.
As noted in \cref{rmk:alg}, it either goes into recursion with at most $m/2$ edges, or executes the same algorithm as in cases not having both in-heavy and out-heavy vertices, with a subset $U \subset V$ of remaining vertices initially.
A similar argument shows that the same inequality is satisfied with an additional $\exp(-64d/D)$ factor.
Thus, increasing $c_2$ by $64/\log c_1$ concludes the proof.
\end{proof}

\paragraph{Edge cutting probability.}

Fix an edge $e=(u,v)$ with $d_e \le \epsilon D/\log\log(mD)$.
It turns out the edge cutting probability can be directly derived from the clustering probability.

\begin{lemma}
\label{lem:prob-cut-cluster}
For edge $e=(u,v) \in E$ with $d_e \le d$, $(u,v)$ is cut only if none of $u,v$ is clustered.
\end{lemma}

\begin{proof}
We prove the lemma by induction.
Observe that $(u,v)$ is cut only if $u,v$ go into the same recursion, in which $(u,v)$ is cut, or if \cref{alg:jason-modified} cuts either a union $B$ of balls or its complement, with exactly one of $u,v$ in $B$.
The former case is handled by induction.

For the latter case, first consider a single out-ball $B=\Bout_G(t,r)$.
By the construction of $\sigma$, to cut $(u,v)$, it must be that $u \in B$ and $v \not\in B$, meaning that $\dist_G(t,u) \le r$ and $\dist_G(t,v) > r$.
Combined with $d_e \le d$, it further implies that $\dist_G(t,u) \ge \dist_G(t,v) - d_e > r-d$ and $\dist_G(t,v) \le \dist_G(t,u) + d_e \le r+d$.
In other words, $\dist_G(t,u),\dist_G(t,v) \in (r-d,r+d]$.
Since \cref{alg:jason-modified} marks all vertices in $\Bout_G(t,r+d) \setminus \Bout_G(t,r-d)$ whenever either $B$ or its complement is cut, $u,v$ must be marked at the same step, rendering both of them not clustered.

For a union of out-balls $B$, which only occurs in case (b) of having both in-heavy and out-heavy vertices, since \cref{alg:jason-modified} marks all vertices $v$ with $\dist_B(v) \in (r-d,r+d]$ whenever $B$ is cut, the same argument as above applies with $\dist_B(\cdot)$.

This concludes the proof as in-balls are symmetric.
\end{proof}

\begin{lemma}
\label{lem:prob-cut}
For edge $e=(u,v) \in E$, $(u,v)$ is cut by \cref{alg:jason-modified} with probability at most $1-\exp(-\O(d_e/D \cdot \log m \log\log(mD)))$.
\end{lemma}

\begin{proof}
As noted in \cref{rmk:alg}, $d$ is used only for clustering, while having nothing to do with cutting.
In other words, the event of $(u,v)$ being cut is independent of the value of $d$.
Suppose \cref{alg:jason-modified} is executed with $d=d_e$ instead.
The above argument shows the probability of $(u,v)$ being cut is unchanged.
By \cref{lem:prob-cut-cluster}, the probability is upper bounded by the probability of $u$ not being clustered, which is at most $1-p(m)$ by \cref{lem:prob-cluster}.
\end{proof}

\paragraph{Separation property.}

As noted in \cref{rmk:alg}, the intuition for the separation property comes from that vertices within distance $d$ to the cutting boundary are marked and excluded from clustering.
Furthermore, we emphasize that it is necessary to mark vertices on both sides of the boundary.
To see this, consider two vertices on the same side of the cut with distance at most $d$ between them.
It is possible that the distance increases to larger than $d$ on the induced subgraph in the recursion, which happens if some vertex on the shortest path is on the other side of the cut.
From the view of the recursion, these two vertices are already separated, while actually they are not in the original graph.
This would fail the attempt to prove the separation property by induction.
Nevertheless, by marking vertices on both sides of the cutting boundary, at least one of the two vertices must be marked in such case.
This is formalized in the following lemma.

\begin{lemma}
\label{lem:dist}
For integers $r,d,k \ge 1$ such that $r \ge d$, and vertices $t_1,\ldots,t_k \in V$, let $B = \bigcup_{i=1}^k \Bout_G(t_i,r)$,  $G'=G[B]$, and $G''=G[V \setminus B]$.
For $u,v \in V$, it holds that:
\begin{enumerate}
    \item If $\dist_B(u) \le r-d$ and $\dist_B(v) > r+d$, then $\dist_G(u,v) \ge 2d$.
    \item If $\dist_B(u),\dist_B(v) \le r-d$ and $\dist_G(u,v) \le d$, then $\dist_{G'}(u,v)=\dist_G(u,v)$.
    \item If $\dist_B(u),\dist_B(v) > r+d$ and $\dist_G(u,v) \le d$, then $\dist_{G''}(u,v)=\dist_G(u,v)$.
\end{enumerate}
\end{lemma}

\begin{proof}
For the first item, note that $\dist_G(t_i,v) \le \dist_G(t_i,u) + \dist_G(u,v)$ for any $i \in [k]$ by triangle inequality.
Minimizing over $i \in [k]$ on both sides, we get $\dist_B(v) \le \dist_B(u) + \dist_G(u,v)$.
Thus, $\dist_G(u,v) \ge \dist_B(v) - \dist_B(u) \ge (r+d) - (r-d) = 2d$.

For the second item, consider the shortest path $P$ from $u$ to $v$ in $G$.
If $\dist_G(u,v) \le d$, by a similar argument, any vertex $x$ on $P$ satisfies $\dist_B(x) \le \dist_B(u) + \dist_G(u,x) \le (r-d) + \dist_G(u,v) \le (r-d) + d = r$.
In other words, $P$ is entirely contained in $G'$ and thereby $\dist_{G'}(u,v)=\dist_G(u,v)$.

For the third item, also consider the shortest path $P$ from $u$ to $v$ in $G$.
If $\dist_G(u,v) \le d$, any vertex $x$ on $P$ satisfies $\dist_B(x) \ge \dist_B(v) - \dist_G(x,v) > (r+d) - \dist_G(u,v) \ge (r+d) - d = r$.
So $P$ is entirely contained in $G''$ and thereby $\dist_{G''}(u,v)=\dist_G(u,v)$.
\end{proof}

\begin{lemma}
\label{lem:sep}
Let $\sigma = (C_1,\ldots,C_k)$ and $\mk(\cdot)$ be returned by \cref{alg:jason-modified}.
For $i,j \in [k]$ such that $i > j$, and vertices $u \in C_i$ and $v \in C_j$ such that $\mk(i)=\mk(j)=0$, it holds that $\dist_G(u,v) > d$.
\end{lemma}

\begin{proof}
We prove the lemma by induction.
In the base case of $m \le 1$, $\sigma$ is induced by a topological order, implying that there is no path from $u$ to $v$.
Otherwise, \cref{alg:jason-modified} constructs $\sigma$ step by step while cutting either unions of balls or their complements.
At each step, we only consider the case of out-balls.
The case of in-balls is symmetric.

First suppose both of $u,v$ are in a union of out-balls $B$.
In order for $u,v$ not to be marked, we must have $\dist_B(u),\dist_B(v) \le r-d$.
If $\dist_G(u,v) \le d$, we get $\dist_{G[B]}(u,v) = \dist_G(u,v) \le d$ by \cref{lem:dist}.
In other words, the shortest path from $u$ to $v$ is entirely contained in $G[B]$.
Regardless of \cref{alg:jason-modified} recursing on either $G[B]$ or $G[U \setminus B]$, by induction, at least one of $u,v$ will eventually be marked.
The same argument applies to the case of $u,v$ both in $U \setminus B$.

If exactly one of $u,v$ is in $B$, observe that \cref{alg:jason-modified} constructs $\sigma$ such that each time recursing on an ``out-component'' (i.e., the subgraph induced by a union of out-balls or complement of a union of in-balls), it concatenates the ordered clustering $\sigma'$ of the out-component to the beginning $\sigma^+$.
Meanwhile, each time recursing on an ``in-component'' (i.e., the subgraph induced by a union of in-balls or complement of a union of out-balls), it concatenates the ordered clustering $\sigma'$ of the in-component to the end $\sigma^-$.
\cref{alg:jason-modified} returns $\sigma^- \circ \sigma^+$ after all recursions.

Altogether, it ensures that any vertex of an out-component occurs after vertices of all in-components and remaining vertices of $U$ when recursing on the component.
Also, any vertex of an in-component occurs before vertices of all out-components and remaining vertices of $U$ when recursing on the component.
As a result, it must be that $u \in B$ and $v \not\in B$ as $B$ is a union of out-balls.
In order for $u,v$ not to be marked, we must have $\dist_B(u) \le r-d$ and $\dist_B(v) > r+d$.
By \cref{lem:dist}, we get $\dist_G(u,v) \ge 2d > d$ as claimed.
\end{proof}

\paragraph{Diameter property.}

\begin{lemma}
Let $\sigma=(C_1,\ldots,C_k)$ be returned by \cref{alg:jason-modified}.
For $i \in [k]$, $C_i$ has diameter at most $D$.
\end{lemma}

\begin{proof}
We prove by induction with two base cases.
One is when $m \le 1$, implying that each $G[C_i]$ is a single vertex.
The other is case (a) of having both in-heavy and out-heavy vertices.
The same argument as in \cite{Li25}'s Lemma 10 shows that $\Bin_G(s,r) \cap \Bout_G(t,r)$ has diameter at most $D$.
\end{proof}

We remark that $C_i$ may not be strongly connected as required by \cref{thm:jason-modified}.
Nevertheless, the diameter property can be satisfied by clustering $C_i$ into strongly connected components and ordering them topologically.
No other property is affected.

\paragraph{Running time.}

To bound the running time, we have to implement vertex marking as efficiently as ball cutting.
This can be achieved by adapting Dijkstra’s algorithm as will be seen in \cref{lem:time}.
We emphasize that \cite{Li25}'s analysis for the straightforward implementation of the algorithm bounds the running time in expectation because
\begin{itemize}
    \item The number $k$ of sampled vertices in each iteration is bounded in expectation.
    \item For fixed $k$, the number of times each vertex is enqueued in Dijkstra’s algorithm is bounded in expectation as well.
\end{itemize}
\cite{Li25} makes the running time bound also hold with high probability by restarting.
However, the same cannot be trivially applied to \cref{alg:jason-modified}.
Indeed, the event of restarting is correlated with $k$ and the random order, while the restarting probability can be as large as a constant.
So it may overwhelm the exponentially small probability guaranteed by \cref{lem:prob-mark}.
Nevertheless, we show that both of the above bounds actually happen with high probability.

\begin{lemma}
\label{lem:time}
\cref{alg:jason-modified} can be implemented in $\O((m+n\log\log n)\log^2 n)$ time with high probability on graphs with polynomially bounded, integral edge lengths.
\end{lemma}

\begin{proof}
We first bound the running time in expectation.
Compared to \cite{Li25}'s algorithm, the only additional running time of \cref{alg:jason-modified} comes from marking vertices.
We focus on cases not having both in-heavy and out-heavy vertices, which are the most time consuming.
It turns out that it can be implemented in the same way as \cite{Li25}'s approach to computing balls, which is based on Dijkstra’s algorithm.

Concretely, we execute Dijkstra’s algorithm for each vertex $v \in S^\pm$ in the sampled random order, to compute $B^\pm_G(v,r^\pm) \cap U$ while marking vertices in $(B^\pm_G(v,r^\pm+d) \setminus B^\pm_G(v,r^\pm-d)) \cap U$.
Throughout, each vertex $u$ maintains the following additional information.
\begin{itemize}
    \item The shortest distance at $u$ for any $v' \in S^\pm$ before $v$ in the random order.  
    \item An indicator of whether $u \in B^\pm_G(v',r^\pm)$ for some $v' \in S^\pm$ before $v$ in the random order.
    \item An indicator of whether $u$ has been marked.
\end{itemize}
All above information can be updated in constant time.
Note that $u$ is marked by $v$ only if $u \not\in B^\pm_G(v',r^\pm)$ for all $v' \in S^\pm$ before $v$ in the random order.
Meanwhile, $u$ is enqueued in Dijkstra’s algorithm only if the shortest distance at $u$ is shortened by $v$.
This is because otherwise, for any vertex $u' \in B^\pm_G(v,r^\pm+d)$ reached via $u$, either $u \in B^\pm_G(v',r^\pm)$ for some $v' \in S^\pm$, or $u$ has been marked.
Thus, the same analysis as in \cite{Li25}'s Lemmas 8 and 13 bounds the running time in expectation.

We now show that the running time actually happens with high probability.
In the following, $m$ always denotes the number of edges in the original graph (as opposed to the induced subgraph for the current recursive instance), and high probability means $1-1/\poly(mD)$.
For each recursive instance and each iteration $i(\pm)$, by Chernoff bound, the number $k$ of sampled vertices is at most $\O(2^{2^i}\log(mD))$ with high probability.
By union bound, with high probability, it holds for all recursive instances and all iterations.

Furthermore, consider fixed vertex $u$ and the random order for sampled vertices.
With probability $1/2$, the first vertex in the random order is among the $k/2$ closest to $u$, so at least $k/2$ sampled vertices will not cause $u$ to be enqueued in Dijkstra’s algorithm.
Similarly, for any sampled vertex that has not been eliminated by previous sampled vertices, it eliminates at least half of remaining sampled vertices with probability $1/2$.
All sampled vertices are eliminated once $\log k=\O(2^i+\log\log(mD))$ such vertices are encountered.
Summing over all iterations, a total of $\O(\log(mD))$ such vertices are encountered.
Since this property holds with probability $1/2$ each time $u$ is enqueued, again by Chernoff bound, $u$ is enqueued a total of $\O(\log(mD))$ times with high probability.
Another union bound over all recursive instances and all vertices concludes the proof as $D=\poly(m)$.
\end{proof}

\bibliographystyle{alpha}
\bibliography{ref}

\appendix

\section{Missing Proofs from \cref{sec:jason-modified}}
\label{append:proof-jason}

\begin{lemma}
\label{lem:avg-ratio-adj}
For reals $k_1,k_2,k_3 > 0$, it holds that
\[
\frac{k_2}{k_1+k_2}+\frac{k_3}{k_1+k_2+k_3} \le 2 - 2\sqrt{\frac{k_1}{k_1+k_2+k_3}}.
\]
Equality is achieved if and only if $k_1+k_2=\sqrt{k_1(k_1+k_2+k_3)}$.
\end{lemma}

\begin{proof}
By AM-GM inequality, we have
\begin{align*}
\frac{k_2}{k_1+k_2}+\frac{k_3}{k_1+k_2+k_3}
& = 1-\frac{k_1}{k_1+k_2}+1-\frac{k_1+k_2}{k_1+k_2+k_3}\\
& \le 2 - 2\sqrt{\frac{k_1}{k_1+k_2} \cdot \frac{k_1+k_2}{k_1+k_2+k_3}}\\
& = 2 - 2\sqrt{\frac{k_1}{k_1+k_2+k_3}}.\qedhere
\end{align*}
\end{proof}

\begin{lemma}
\label{lem:avg-ratio-prefix-1}
For integer $r > 1$, and reals $k,k_1,\ldots,k_r \ge 0$ such that $k_1 \ge 1$ and $k_1+\ldots+k_r=k$, it holds that
\[
\frac{1}{r} \cdot \sum_{i=1}^r \frac{k_i}{\sum_{j=1}^i k_i} \le 1-\left(1-\frac{1}{r}\right) \cdot k^{-\frac{1}{r-1}}.
\]
\end{lemma}

\begin{proof}
We claim that for fixed $r$ and $k$, the maximum possible value of LHS is achieved when $k_{\le i} = n^{(i-1)/(r-1)}$.
To see this, consider any possible $k_1,\ldots,k_r$.
Suppose there exists $i \in (1,r)$ such that $k_{\le i} \ne \sqrt{k_{\le (i-1)}k_{\le (i+1)}}$.
By \cref{lem:avg-ratio-adj}, setting $k_{\le i}=\sqrt{k_{\le (i-1)}k_{\le (i+1)}}$, which is same as appropriately setting $k_i$ conditioned on fixed $k_i+k_{i+1}$, strictly increases the value of LHS. 
Consequently, the value of LHS is maximized only if $k_{\le i} = \sqrt{k_{\le (i-1)}k_{\le (i+1)}}$ for all $i \in (1,r)$.
In other words, $k_{\le i} = k_1 \cdot (k/k_1)^{\frac{i-1}{r-1}}$ for $i \in [r]$, implying that
\[
\textrm{LHS} = \frac{1}{r} \cdot \left(1 + \sum_{i=2}^r \frac{k_i}{k_{\le i}}\right) = 1 - \frac{1}{r} \cdot \sum_{i=2}^r \frac{k_{\le (i-1)}}{k_{\le i}} = 1-\left(1-\frac{1}{r}\right) \cdot \left(\frac{k}{k_1}\right)^{-\frac{1}{r-1}}.
\]
This concludes the proof by setting $k_1=1$.
\end{proof}

\begin{proof}[Proof of \cref{lem:avg-ratio-prefix-d}]
Observe that
\[
\textrm{LHS} = \frac{1}{r} \cdot \sum_{\hat{i}=1}^{d} \sum_{i=0}^{\lfloor (r-\hat{i})/d \rfloor} \frac{\sum_{j = \max(\hat{i} + (i-1) \cdot d + 1, 1)}^{\hat{i} + i \cdot d} k_j}{\sum_{j=1}^{\hat{i} + i \cdot d} k_j}.
\]
For fixed $\hat{i} \in [d]$, each term of the inner summation corresponds to a unique interval in $[1,\hat{i}],[\hat{i}+1,\hat{i}+d],\ldots,[\hat{i} + (\hat{r}-1) \cdot d, \hat{i} + \hat{r} \cdot d]$, where $\hat{r} = \lfloor (r-\hat{i})/d \rfloor$.
Let $\hat{k} = k_{\le (\hat{i} + \hat{r} \cdot d)}$.
Applying \cref{lem:avg-ratio-prefix-1}, we get that the inner summand for $\hat{i}$ is at most
\[
(\hat{r}+1) \cdot \left(1 - \left(1-\frac{1}{\hat{r}+1}\right) \cdot \hat{k}^{-\frac{1}{\hat{r}}}\right) \le (\hat{r}+1) \cdot \left(1 - \left(1-\frac{1}{\lfloor r/d \rfloor}\right) \cdot k^{-\frac{1}{\lfloor r/d \rfloor - 1}}\right),
\]
because $\hat{r}+1 \ge \lfloor r/d \rfloor$ and $\hat{k} \le k$.
Combining all above concludes the proof as the summation of $\hat{r}+1$ over $i \in [d]$ is just $r$.
\end{proof}
\end{document}